\documentclass[11pt]{article}

\usepackage{
    amssymb,
    amsmath,
    amsfonts,
    eurosym,
    geometry,
    ulem,
    graphicx,
    caption,
    color,
    setspace,
    sectsty,
    comment,
    footmisc,
    caption,
    pdflscape,
    subcaption,
    subfiles,
    titling,
    array,
    hyperref,
    booktabs,
    longtable,
    float,
    authblk,
    threeparttable,
    makecell}

\usepackage{relsize}

\usepackage[citestyle=ama, bibstyle=ama]{biblatex}

\AtEveryBibitem{%
  \clearfield{note}%
}
\AtEveryBibitem{\clearlist{publisher}}
\AtEveryBibitem{\clearlist{language}}

\usepackage{pgf,tikz}
\usetikzlibrary{arrows, automata}
\usetikzlibrary{shapes.geometric,positioning}
\usetikzlibrary{positioning,calc}

\usepackage{siunitx}
\newcolumntype{d}{S[input-symbols = ()]}

\newtheorem{theorem}{Theorem}

\newenvironment{proof}[1][Proof]{\noindent\textbf{#1.} }{\ \rule{0.5em}{0.5em}}

\newcolumntype{L}[1]{>{\raggedright\let\newline\\arraybackslash\hspace{0pt}}m{#1}}
\newcolumntype{C}[1]{>{\centering\let\newline\\arraybackslash\hspace{0pt}}m{#1}}
\newcolumntype{R}[1]{>{\raggedleft\let\newline\\\arraybackslash\hspace{0pt}}m{#1}}
\DeclareMathOperator{\E}{E}

\begin{document}

\begin{titlepage}
\title{Estimating and evaluating counterfactual prediction models} 
\author[1-3]{Christopher B. Boyer\thanks{Corresponding author email: \href{mailto:boyerc5@ccf.org}{boyerc5@ccf.org}}}
\author[3-6]{Issa J. Dahabreh}
\author[7]{Jon A. Steingrimsson}

\affil[1]{Department of Quantitative Health Sciences, Cleveland Clinic Research, Cleveland, OH.}
\affil[2]{Department of Medicine, Cleveland Clinic Lerner College of Medicine, Case Western Reserve University, Cleveland, OH.}
\affil[3]{Department of Epidemiology, Harvard T.H. Chan School of Public Health, Boston, MA.}
\affil[4]{CAUSALab, Harvard T.H. Chan School of Public Health, Boston, MA.}
\affil[5]{Department of Biostatistics, Harvard T.H. Chan School of Public Health, Boston, MA.}
\affil[6]{Richard A. and Susan F. Smith Center for Outcomes Research, Beth Israel Deaconess Medical Center, Boston, MA.}
\affil[7]{Department of Biostatistics, Brown University School of Public Health, Providence, RI.}
\date{\today}
\maketitle

\begin{abstract}
\noindent

Counterfactual prediction methods are required when a model will be deployed in a setting where treatment policies differ from the setting where the model was developed, or when a model provides predictions under hypothetical interventions to support decision-making. However, estimating and evaluating counterfactual prediction models is challenging because, unlike traditional (factual) prediction, one does not observe the potential outcomes for all individuals under all treatment strategies of interest. Here, we discuss how to estimate a counterfactual prediction model, how to assess the model's performance, and how to perform model and tuning parameter selection. We provide identification and estimation results for counterfactual prediction models and for multiple measures of counterfactual model performance, including loss-based measures, the area under the receiver operating characteristic curve, and the calibration curve. Importantly, our results allow valid estimates of model performance under counterfactual intervention even if the candidate prediction model is misspecified, permitting a wider array of use cases. We illustrate these methods using simulation and apply them to the task of developing a statin-na\"{i}ve risk prediction model for cardiovascular disease. \\

\vspace{0in} 

\bigskip
\end{abstract}
\setcounter{page}{0}
\thispagestyle{empty}
\end{titlepage}
\pagebreak \newpage

\doublespacing

\section{Introduction} \label{sec:introduction}
Many common tasks in prediction modeling involve ``what if'' questions that call for predictions under conditions that are contrary to those prevalent when the training data are collected. When these changing conditions can be usefully thought of as resulting from  hypothetical interventions on variables relevant to the outcome, we refer to the task as counterfactual prediction \cite{hernan_second_2019,dickerman_counterfactual_2020,van_geloven_prediction_2020}. There exist several broad categories of modeling tasks that involve counterfactual predictions. 

First, a model may be deployed in a setting that differs from the setting from which training data were gathered in terms of patterns of post-baseline treatments. These changes in treatment patterns could occur in a given target population over time due to treatment policy changes, commonly referred to as problems of ``domain adaption'' or ``dataset shift''  \cite{finlayson_clinician_2021, subbaswamy_development_2020}. Alternatively, they could occur when a model is transported from a source population to a new target population that has a similar baseline covariate distribution but differs in terms of post-baseline treatment patterns. In both cases, differences between the training and deployment settings can cause the predictive performance of models to degrade. While, ideally, one would re-train the model using data in the new setting, to reduce data collection costs or as a stop-gap while new data are collected, investigators may attempt to use existing data to adapt or tailor the model to target the expected outcome under a hypothetical intervention to match treatment patterns in the target setting where the model will be applied. Alternatively, investigators may use existing data to evaluate the counterfactual performance of the model, unadapted, in the target setting \cite{pajouheshnia_accounting_2017}. Both are examples of counterfactual prediction tasks.

Second, a model may be intended for use as a decision support tool, for instance, by clinicians to counsel patients about their risk of disease under alternative treatment plans. In certain unrealistic circumstances, such as when using data from a trial that (1) is representative of the target population, (2) randomly assigned participants to the treatment strategies considered for decision support, and (3) had complete adherence to the assigned strategies and no loss to follow-up \cite{glasziou_evidence_1995,dahabreh_using_2016,kent_personalized_2018,kent_predictive_2020,hoogland_tutorial_2021}, obtaining prediction models for the outcome under different treatment strategies may not require the additional formalism of counterfactual prediction. However, virtually all practical applications diverge from this ideal: trial participants are not representative of the target population, patterns of adherence differ between the trial and the setting in which the model will be deployed, followup is incomplete, or the model is trained using observational data \cite{schulam_reliable_2017-1,dickerman_counterfactual_2020}. Therefore, addressing the research question entails contrary-to-fact interventions, and calls for using methods for counterfactual prediction.

Third, specific interventions in the training data may be viewed as pernicious, undesired, or incompatible with the target setting in which the model will be deployed and thus ``removal'' of their effects may be attempted via prediction under a hypothetical intervention that would eliminated them. For instance, an analyst may be interested in developing a model for the risk of death in the absence of surgery, but may only have data collected from an observational setting where some of the patients ultimately received surgery, a problem sometimes described as ``treatment drop-in'' \cite{van_geloven_prediction_2020, sperrin_using_2018}. Provided that the hypothetical elimination of these interventions is sensible in the target setting and supported by the training data (in the sense that at least some individuals in the training data in fact follow the strategies of interest), attempting to estimate the treatment-free or treatment-na\"{i}ve risk is another example of a counterfactual prediction task.

Compared to factual prediction tasks, counterfactual prediction tasks require stronger assumptions that are closely related to those commonly invoked in the causal inference literature. They also require alternative methods that allow analysts to estimate and draw inferences about counterfactual estimands. Recently, several methods have been proposed for counterfactual prediction that seek to re-tool existing causal inference approaches for use in the prediction setting \cite{lin_scoping_2021,sperrin_using_2018,dickerman_predicting_2022,schulam_reliable_2017-1}. However, these methods do not always distinguish between the covariates available at the time the prediction is to be made and those required to satisfy the assumptions necessary for identifiability of the counterfactual estimand \cite{coston_counterfactual_2021}. Additionally, there has been comparatively little work to date (one exception is \textcite{coston_counterfactual_2020}) on how to estimate the performance of counterfactual prediction models independently of the method used to fit the model while also allowing for the possibility that the model is misspecified. 

In this paper, we provide formal identifiability results for estimating a counterfactual prediction model. We derive and describe estimation methods for the case when a sufficient set of covariates is available during model training to control treatment-outcome confounding, but allow for the model itself to be conditional on a smaller subset of ``predictors'' that are available when deployed. We also provide identifiability results and novel estimators for multiple model performance measures including loss-based measures (such as the mean-squared error), the area under the receiver operating characteristics curve (AUC), and the calibration curve. We illustrate these methods using simulation and apply them to the task of developing a statin-na\"{i}ve risk prediction model for cardiovascular disease.

\section{Set up and notation} \label{sec:setup}
Let $Y$ be the outcome of interest, $X$ a vector of baseline covariates, and $A$ an indicator of treatment after baseline. We assume that data are independent realizations $\{(X_i, A_i, Y_i)\}_{i=1}^n$ from a source population in an observational setting in which treatment is not determined by the investigator but rather initiated according to the pattern $f^{obs}(A | X)$. The covariates in $X$ include a set sufficient to control confounding of the treatment-outcome relationship ($L$) as well as additional predictors of the outcome ($P$). An example causal directed acyclic graph (DAG) for this process is shown in Figure \ref{fig:dag1}. We allow for the possibility that only a subset of covariates $X^*$, chosen based on their availability and predictive potential rather than on whether they control confounding, will be use at the time predictions are to be made (note that $X^*$ can include components of both $L$ and $P$).

\begin{figure}[t]
    \centering
    \begin{tikzpicture}[> = stealth, shorten > = 1pt, auto, node distance = 2.5cm, inner sep = 0pt,minimum size = 0.5pt, very thick]
    \tikzstyle{every state}=[
      draw = white,
      fill = white
    ]
    \node[state] (l0) {$L$};
    \node[state] (a0) [right of=l0] {$A$};
    \node[state] (y1) [right of=a0] {$Y$};
    \node[state] (u0) [below of=l0] {$U$};
    \node[state] (p0) [right of=u0] {$P$};

    \path[->] (l0) edge node {} (a0);
    \path[->] (l0) edge [out=45, in=135] node {} (y1);

    \path[->] (a0) edge node {} (y1);
    
    \path[->] (p0) edge [out=0,in=255] node {} (y1);

    \path[->] (u0) edge node {} (y1);
    \path[->] (u0) edge node {} (l0);
    \path[->] (u0) edge node {} (p0);
    \end{tikzpicture}
    \caption{Example causal directed acyclic graph (DAG) for prediction in a setting with a single time fixed treatment $A$ over follow up.}
    \label{fig:dag1}
\end{figure}
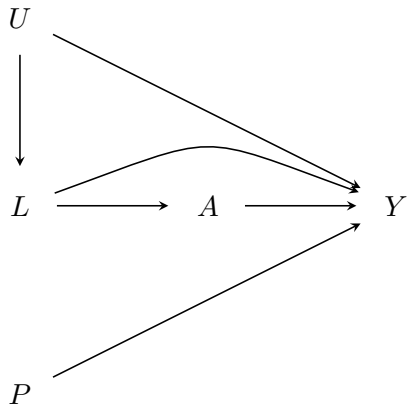

Let $Y^a$ denote the potential outcome under a hypothetical intervention to set treatment $A$ to $a$. Our objective is to build a model for the conditional expectation of the potential outcome in the target population under intervention $f^*(A|X)$ and assess its performance, where here we use $f(\cdot)$ generically to denote a density. For simplicity, in the main text we focus on simple intervention strategies where everyone receives treatment $A = a$, in which case the goal is to estimate $E[Y^a | X^*]$. Similarly, we assume, for now, that treatment is time-fixed, that is treatment, when initiated, is done so immediately after baseline and sustained over follow up or alternatively that its effect is independent of the time on treatment. In the Appendix, we extend our results to more general intervention strategies, including random and dynamic treatment regimes (Section \ref{sec:randomdynamic}), and to the case where treatment is time-varying (Section \ref{sec:timevarying}). Finally, we assume no censoring or loss to follow up. 

For the purposes of assessing model performance, we assume data are randomly split into a training set of $n_{train}$ observations and a test set of $n_{test}$ observations, with $n = n_{train} + n_{test}$. Let $D_{train}$ and $D_{test}$ be indicators of whether an observation is in the training set or test set respectively. We assume the model is to be trained in the training set and evaluated in the test set. Our results, however, apply equally to existing prediction models or (directly measured) biomarkers, in which case splitting is not necessary. To simplify exposition, we refer to a single test-train split throughout but, as we discuss further in Section \ref{sec:selection}, our set up easily accommodates more complex evaluation strategies such as cross-validation or methods based on bootstrapping. 

\section{Estimands for model estimation and performance} \label{sec:targets}
We distinguish between two possible counterfactual prediction tasks (i) estimating a counterfactual prediction model and (ii) assessing the performance of an arbitrary model, whether tailored or not tailored to a specific setting, through comparison of model predictions against the potential outcome under the specified intervention. For the first task, the goal is to directly estimate the expected potential outcome under intervention $A = a$ conditional on covariates $X^*$ which are a subset of $X$, in other words the estimand is $\mu_a(X^*) \equiv \E[Y^a | X^*]$. As $\mu_a(X^*)$ depends on the unknown potential outcome $Y^a$, it is not a function of the observed data without additional assumptions. Therefore, we derive identifiability results for $\mu_a(X^*)$ and associated estimation procedures to inform model development. We refer to this task as ``model tailoring'' because the estimation procedure is tailored to the outcome mean in a target population under the specified intervention.

To evaluate the performance of a prediction model, we would like to estimate one or more performance measures.  A number of performance measures have been proposed in the prediction literature \cite{harrell_multivariable_1996, altman_what_2000, steyerberg_clinical_2019}. An example performance measure of interest is 
\begin{equation*}
    \psi(a) \equiv \E[(Y^a - \widehat{\mu}(X^*))^2 \mid D_{test} = 1]
\end{equation*}
where the squared error loss $(Y^a - \widehat{\mu}(X^*))^2$ quantifies the discrepancy between the potential outcome under treatment level $A = a$ and the model prediction $\widehat{\mu}(X^*)$ in terms of the squared difference. In the main text, we focus on the mean squared error as the performance measure of interest $\psi(a)$. In the Appendix, we extend our results to the case where the measure is any member of a generic class of loss functions for counterfactual outcomes $L(Y^a,  \widehat{\mu}(X^*))$ (Section \ref{sec:tf_proof}) as well as more complex risk-based measures such as the area under the receiver operating characteristics curve (AUC, Section \ref{sec:auc}), which depends on paired observations, and the calibration curve (Section \ref{sec:calib}), which is a functional. Importantly, all the measures considered are identifiable without assuming that the model $\mu(X^*)$ is correctly specified, in the sense that it need not converge to the true conditional expectation $\mu_a(X^*)$.

\section{Identifiability conditions} \label{sec:identifiability}
We assume the following identifiability assumptions which have been described in more detail elsewhere \cite{hernan_causal_2020, robins_new_1986, robins_graphical_1987}.

\begin{enumerate}
    \item[A1.] \textit{Consistency.} If $A = a$, then $Y^a = Y$ 
    \item[A2.] \textit{Conditional exchangeability.} $Y^a \perp\!\!\!\perp A \mid X$ 
    \item[A3.] \textit{Positivity.} For all $x$ with positive density, i.e. $f_X(x) > 0$, $\Pr[A = a \mid X = x] > 0$ 
\end{enumerate}

Consistency implies that observed outcomes among those with $A = a$ reflect potential outcomes under corresponding level of treatment. It would be violated if, for instance, there were interference, Hawthorne effects, or if there were multiple ``hidden'' versions of the treatment under consideration \cite{rubin_randomization_1980,dahabreh2019generalizing}. The latter may be particularly likely in prediction datasets where information on treatment is passively collected, for instance in a electronic healthcare record where they are not well defined or include hidden components. The exchangeability condition stipulates that treatment is conditionally independent of the potential outcome given covariates $X$. It is often referred to as the ``no unmeasured confounding'' assumption as it implies that there are no unmeasured variables that affect both treatment assignment and the outcome. We require independence conditional on $X$ but recall that the prediction model itself may be conditional on only a subset $X^*$. This allows, for instance, a high-dimensional covariate vector to be collected for model training, while the ultimate prediction tool may only be based on a handful of easy-to-obtain measurements at runtime. In contrast to a typical observational causal inference set up, we allow $X$ to include predictors that are not confounders (e.g., elements of $P$ instead of $L$ in DAG in Figure \ref{fig:dag1}). However, we require that conditioning on these predictors must not lead to a violation of Assumption A2, as would happen if the predictor were in fact a collider. Therefore subject-matter expertise should still guide the selection of the full set of covariates $X$. Finally, the positivity condition implies that there is a positive probability of observed treatment level $A = a$ in all strata of $X$ that have positive density in the population in the setting of interest. Positivity violations may also be common in prediction datasets where data are often assembled for multiple prediction tasks and without regard to who is eligible for treatment.

\section{Estimating a counterfactual prediction model} \label{sec:model}

As we show in Appendix Section \ref{sec:model_proof}, under the conditions above $\E[Y^a | X^*]$ is identified by
\begin{equation}\label{eqn:estimand1}
    \mu_a(X^*) = \E[\E[Y \mid X, A = a, D_{train} = 1] \mid X^*, D_{train} = 1]
\end{equation}
or, equivalently, using an inverse probability weighted expression 
\begin{equation}\label{eqn:estimand2}
    \mu_a(X^*) = \E\left[\frac{I(A = a)}{\Pr[A = a \mid X, D_{train} = 1]} Y \Big| X^*, D_{train} = 1\right]
\end{equation}
The two expressions for $\mu_a(X^*)$ suggest possible approaches for estimating the counterfactual prediction model from the training data. 

One approach, based on equation \ref{eqn:estimand1}, is to subset to participants with corresponding treatment level $A = a$ in the training data and estimate a model $\mu(X)$ for the observed $Y$ conditional $X$, i.e. $\E[Y | X, A = a] = \mu(X)$. Then when the desired predictors $X^*$ are a subset of $X$, the covariates sufficient to ensure exchangeability,  predictions are marginalized (standardized) over the covariates in $X$ that are not in $X^*$. When the dimension of $X^*$ is small, this can be done nonparametrically; however, when $X^*$ is higher-dimensional, an additional modeling step will be required either (1) modeling the estimated $\widehat{\mu}(X)$ as a function of $X^*$, i.e. $\E[\widehat{\mu}(X) | X^*]$ or (2) modeling the conditional density of $X$ given $X^*$, i.e. $f(X | X^*)$.  The resulting predictions will be consistent for $\E[Y^a | X^*]$ provided all models are correctly specified. An alternative suggested in \cite{dickerman_predicting_2022}, would be to simulate samples from the model $\mu(X)$ and fit a second stage model $\mu^*(X^*)$ using only the subset of the predictors of interest. This is similar to the first standardization method above and has the advantage that the second stage model may be developed using an outcome with same support as the original $Y$.

A second approach, based on equation \ref{eqn:estimand2}, is to fit a weighted model $\mu(X^*)$, using for instance weighted maximum likelihood, with weights equal to the probability of receiving treatment level $A = a$ conditional on covariates $X$ necessary to ensure exchangeability, i.e., sample analogs of $W(a) = \frac{I(A = a)}{\Pr[A = a \mid X, D_{train} = 1]}$. This is the basis for previously proposed methods for counterfactual prediction based on inverse probability of treatment weighting \cite{sperrin_using_2018, van_geloven_prediction_2020}. Note that, as before, it is possible to specify a subset of predictors $X^*$ used in the prediction model $\mu(X^*)$ as compared to the full set of covariates $X$ required for exchangeability which are only necessary for defining the weights $W$, however a second marginalization or simulation step is not required. This means estimating a counterfactual prediction model using the weighting approach can be accomplished more easily using off-the-shelf software.

\section{Assessing model performance} \label{sec:performance}

Using the same identifiability conditions, in Appendix Section \ref{sec:tf_proof}, we show the model performance measure $\psi(a)$ is identifiable using data from the test set through the expression
\begin{equation}\label{eqn:mse1}
    \psi(a) \equiv \E\left[\E[(Y - \widehat{\mu}(X^*))^2 \mid X, A=a, D_{test} = 1] \mid D_{test} = 1\right]
\end{equation}
or, equivalently, using an inverse probability weighted expression, 
\begin{equation}\label{eqn:mse2}
    \psi(a) = \E\left[\frac{I(A = a)}{\Pr[A = a \mid X, D_{test} = 1]}(Y - \widehat{\mu}(X^*))^2 \mid D_{test} = 1\right]
\end{equation}
regardless of the model $\widehat{\mu}(X^*)$ (e.g.,~regardless of whether it has been tailored to target $\E[Y^a \mid X]$ or if it is correctly specified). 

As before, the two expressions suggest two different approaches for the estimation of model performance using the test data alone. 

First, using the sample analog of expression (\ref{eqn:mse1}), an estimator of the target MSE is 
\begin{equation}\label{eqn:cl_estimator}
    \widehat{\psi}_{CL}(a) = \frac{1}{n_{test}} \sum_{i=1}^nI(D_{test, i} = 1)\widehat{h}_a(X_i)
\end{equation}
where $\widehat{h}_a(X)$ is an estimator for the conditional loss $\E[(Y - \widehat{\mu}(X^*))^2 \mid X, A=a, D_{test} = 1]$ and from the perspective of estimating $\psi(a)$ may be considered a nuisance function. To keep notation simple, we suppress the dependency of $\widehat{h}_a(X)$ on $\widehat{\mu}$. When the dimension of $X$ is small it may be possible to estimate $\widehat{h}_a(X)$ nonparametrically. In almost all practical cases though some form of modeling will be required; in these cases, $\widehat{\psi}_{CL}(a)$ is a consistent estimator for $\psi(a)$ as long as the model for $\widehat{h}_a(X)$ is correctly specified.

Next, using the sample analog of expression (\ref{eqn:mse2}), an alternative weight-based estimator of the target MSE is 
\begin{equation}\label{eqn:ipw_estimator}
    \widehat{\psi}_{IPW}(a) = \frac{1}{n_{test}} \sum_{i=1}^n \frac{I(A_i = a, D_{test, i} = 1)}{\widehat{e}_a(X_i)}(Y_i - \widehat{\mu}(X^*_i))^2
\end{equation}
where $\widehat{e}_a(X)$ is another nuisance function estimating the probability of receiving treatment level $A = a$ conditional on $X$, i.e. $\Pr[A = a \mid X, D_{test} = 1]$. Again, when the dimension of $X$ is small it may be possible to use the sample analog of $\widehat{e}_a(X)$, but in most cases, it will have to be modeled. In these cases, $\widehat{\psi}_{IPW}(a)$ is a consistent estimator of $\psi(a)$ as long as the model for $\widehat{e}_a(X)$ is correctly specified.

The conditional loss estimator (\ref{eqn:cl_estimator}) relies on correctly specifying the model for the conditional loss and the weighting estimator (\ref{eqn:ipw_estimator})  relies on correctly specifying the model for the probability of treatment. In some settings, one estimator may be preferred over the other: when more is known about the mechanism for ``assigning'' treatment or when the outcome is very rare, the weighting estimator may be preferred \cite{robins_estimating_1992,braitman_rare_2002}; when the process that gives rise to the outcomes is well understood the conditional loss estimator may be preferred. In practice, however, both models may be difficult to specify correctly. Using data-adaptive and more flexible machine learning estimators for estimation of these nuisance models offers the possibility of capturing more complex data generation processes. These data-adaptive estimators generally have slower rates of convergence than the $\sqrt{n}$ rates of parametric models and therefore will not yield asymptotically valid confidence intervals \cite{chernozhukov_doubledebiased_2018}. To address this challenge, we can use a doubly-robust estimator which combines models for the conditional loss $\widehat{h}_a(X)$ and the probability of treatment $\widehat{e}_a(X)$, such as
\begin{equation}
    \widehat{\psi}_{DR}(a) = \frac{1}{n_{test}} \sum_{i=1}^n I(D_{test,i} = 1) \left[ \widehat{h}_a(X_i) + \frac{I(A_i = a)}{\widehat{e}_a(X_i)} \left\{ (Y_i - \widehat{\mu}(X^*_i))^2 - \widehat{h}_a(X_i)\right\}\right]
\end{equation}
As we show in Appendix section \ref{sec:dr}, under mild regularity conditions \cite{robins_higher_2008}, this estimator will be consistent if one of $\widehat{h}_a(X)$ and $\widehat{e}_a(X)$ is correctly specified. It also permits the use of machine learning or data-adaptive estimators that converge at rate slower than $\sqrt{n}$, thus allowing for more flexible estimation of the nuisance functions. This is due to the fact that the empirical process terms governing the convergence of $\widehat{\psi}_{DR}$ involve a product of the estimation errors for $\widehat{h}_a(X)$ and $\widehat{e}_a(X)$ which converge under the weaker condition that only the \textit{combined} rate of convergence for both nuisance functions is at least $\sqrt{n}$ \cite{chernozhukov_doubledebiased_2018}. In Appendix section \ref{sec:dr}, we present results for general loss functions under static interventions. We leave the derivation of doubly robust estimators for the AUC and calibration curve as well as random and dynamic regimes for future work.

\section{Model and tuning parameter selection} \label{sec:selection}

Up to this point, we have assumed that $\mu(X^*)$ is a pre-specified model and ignored any form of model selection (e.g., variable or other specification search) or data-adaptive tuning parameter selection, which may be the case when using an existing (validated) model. In many cases, however, analysts have to select between multiple models or perform a data-adaptive search through a parameter space for tuning parameter selection when developing a prediction model \cite{steyerberg_clinical_2019}. To avoid overfitting, analysts typically use methods such as cross-validation or the bootstrap to perform selection. These techniques rely on optimizing some measure of model performance, such as the MSE.

When performing model or tuning parameter selection for counterfactual prediction, the results from the previous sections suggest that the model performance measure should be targeted to the counterfactual performance if the intervention of interest were universally applied. For example, when using cross-validation for model selection the analyst splits the data into $K$ mutually exclusive ``folds'' and fits the candidate models using $K - 1$ of the folds and estimates the performance of each in the held out fold. This process is repeated $K$ times where each fold is left out once. The final performance estimate is the average of the $K$ estimates and the model with best overall performance is selected (or, alternatively, the tuning parameter with the best performance). When performing counterfactual prediction, at each stage in the procedure the analyst should use modified performance measures such as those in section \ref{sec:performance} above. Failure to do so, can lead to sub-optimal selection with respect to the counterfactual prediction of interest. 


\section{Simulation experiments} \label{sec:simulation}
We performed two Monte Carlo simulation experiments to illustrate (i) the benefits of tailoring models to the correct counterfactual estimand of interest, (ii) the potential for bias when using na\"{i}ve estimators of model performance such as the MSE, (iii) the importance of correct specification of the nuisance models when estimating counterfactual performance, and (iv) the properties of the doubly-robust estimator under misspecification of the nuisance models. We adapt data generation processes previously used for transporting models between settings under covariate shift \cite{steingrimsson_transporting_2023, morrison_robust_2022}.

\subsection{Experiment 1}
We simulated treatment initiation at baseline based on the logistic model $\Pr[A=1 \mid X]=\operatorname{expit} (1.5-0.3 X)$, where predictors $X$ are drawn from $X \sim$ Uniform $(0,10)$. Under this model, about 50\% initiate treatment but those with higher values of $X$ are less likely to start treatment than those with lower values of $X$. We then simulated the outcome using the linear model $Y=1+X+0.5 X^2- 3A + \varepsilon$, where $\varepsilon \sim \mathcal{N}(0, X)$.  We set the total sample size to 1000 and the data were randomly split in a 1:1 ratio into a training and a test set. In this simulation, $X = X^*$. The full process may be written:
\begin{align*}
    X & \sim \text{Unif}(0, 10) \\
    A & \sim \text{Bernoulli}\{\operatorname{expit}(-1.5 + 0.3 X)\} \\
    Y & \sim \text{Normal}(1 + X + 0.5 X^2 - 3 A, X)
\end{align*}    
 
Our goal was to estimate a model and evaluate it's performance in the same population where, contrary to fact, treatment was universally withheld, i.e. we targeted $\E[Y^{a=0} \mid X]$. Note that, for simplicity, in this case predictor variables are the same set necessary to control confounding $X = X^*$. Under this data generating mechanism, the MSE under no treatment is larger than the MSE under the natural course and identifiability conditions 1-3 are satisfied. We considered two parametric model specifications $\mu(X^*; \beta)$ and $\widetilde{\mu}(X^*, \beta)$:
\begin{enumerate}
    \item a correctly specified linear regression model that included the main effects of $X$ and $X^2$, i.e. $\mu(X; \beta) = \beta_0 + \beta_1 X + \beta_2 X^2$.
    \item a misspecified linear regression model that only included the linear effect of $X$, i.e. $\widetilde{\mu}(X; \beta) = \beta_0 + \beta_1 X$.
\end{enumerate} 
For each specification, we also considered two estimation strategies: one using ordinary least squares regression (OLS) and ignoring treatment initiation and the other using weighted least squares regression (WLS) where the weights were equal to the inverse of the probability of being untreated. As discussed above the latter specifically targets the counterfactual estimand under no treatment. Finally, we considered two approaches for estimating the performance of the models in the test set: a na\"{i}ve estimate of the MSE using observed outcome values, i.e., 
$$\widehat{\psi}_{Na\ddot{i}ve}(a) = \frac{1}{n_{test}} \sum_{i=1}^n I(D_{test,i} = 1) (Y_i - \widehat{\mu}(X_i))^2,$$ 
and the inverse-probability weighted estimator $\widehat{\psi}_{IPW}(a)$ from section \ref{sec:performance}. For the latter, we estimate a correctly specified logistic regression model for $e_a(X)$, i.e. $e_a(X) = \operatorname{expit}(\alpha_0 + \alpha_1 X)$, in the test set to estimate the weights. Lastly, we also calculated the ``true'' MSE under intervention to withold treatment by generating test data under same process as above but setting the treatment $A$ to $a = 0$ for everyone and then averaging across simulations.

\begin{table}[t]
    \centering
    \caption{Simulation results for experiment 1 comparing the performance of OLS and WLS models using the na\"{i}ve and inverse probability weighting (IPW) estimators of the MSE.}
    \begin{threeparttable}
        \begin{tabular}{p{3cm}R{1.25cm}R{1.25cm}R{1.25cm}}
        \toprule
        Model $\mu(X)$ & Na\"{i}ve & IPW & Truth  \\
        \midrule
        Correct & & & \\
        \hspace{1em}OLS & 2.9 & 3.6 & 3.6\\
        \hspace{1em}WLS & 5.5 & 1.0 & 1.0\\
        \addlinespace[0.25em]
        Misspecified & & & \\
        \hspace{1em}OLS & 16.8 & 17.5 & 17.5\\
        \hspace{1em}WLS & 19.5 & 15.0 & 15.0\\
        \bottomrule
        \end{tabular}
        \begin{tablenotes}
        \item \noindent Correct and misspecified refers to the specification of the  prediction model $\mu(X)$. OLS = model estimation using ordinary least squares regression (unweighted); WLS = model estimation using weighted least squares regression with weights equal to the inverse probability of being untreated. Results were averaged over 10,000 simulations. The true counterfactual MSE was obtained using numerical methods. 
        \end{tablenotes}
        \end{threeparttable}
    
\end{table}

Table 1 shows the results of the experiment based on 10,000 simulations. In general, correctly specified models yielded smaller average MSE than misspecified models. Comparing the performance of OLS and WLS estimation, when using $\widehat{\psi}_{Na\ddot{i}ve}(a)$, the na\"{i}ve estimator of the MSE, OLS seemed to produce better predictions than WLS when correctly specified (average MSE of 2.9 vs. 5.5) as well as when misspecified (average MSE of 16.8 vs. 19.5). In contrast, when using $\widehat{\psi}_{IPW}(a)$, the inverse-probability weighted estimate of the MSE, WLS performed better than OLS both when the model was correctly specified (average MSE of 1.0 vs. 3.6) and when misspecified (average MSE of 15.0 vs. 17.5). For reference, in the last column we show the true counterfactual MSE that would be obtained if one had access to the potential outcomes (obtained via numerical methods). We found that the average of the inverse probability weighted estimator across the simulations was equivalent to this quantity for all specifications and for both OLS and WLS estimation. This suggests that only the modified estimators of model performance in section \ref{sec:performance} are able to accurately estimate the counterfactual performance of the model. Indeed, under this data generation process, if one were to use the na\"{i}ve estimator one might erroneously conclude that the OLS model is the better choice.

\subsection{Experiment 2}

In the previous experiment, we assumed the nuisance models for the MSE were correctly specified. We now consider estimation of model performance in the more likely case that nuisance models are misspecified. Using the results from Section \ref{sec:auc} in the Appendix, we also estimate the area under the receiver operating characteristics curve (AUC). We simulated treatment initiation over follow up $A$ based on the logistic model $\operatorname{Pr}[A=1 \mid X]=\operatorname{expit}(0.5 - 2 X_1 + 3 X_1^2 + 2 X_2 - X_3)$, where $X$ is now a vector of predictors drawn from a 3-dimensional multivariate normal with mean vector $\mu = (0.2, 0, 0.5)$ and covariance matrix $\Sigma = \text{diag}(0.2, 0.2, 0.2)$. This resulted in expected treatment initiation over follow up of 55\%. We also simulated a binary outcome from a Bernoulli distribution with mean $\operatorname{expit}(0.2 + 3 X_1 - 2 X_1^2 + 2 X_2 + X_3 - 2 A)$, implying an average probability of the outcome of 66\% among the untreated and 32\% among treated. Again, we set the total sample size to 1000 and randomly split the data in a 1:1 ratio into a training and a test set. 
\begin{align*}
    X & \sim \text{MVN}(\mu, \Sigma) \\
    A & \sim \text{Bernoulli}\left\{\text{expit}\left(0.5 - 2 X_1 + 3 X_1^2 + 2 X_2 - X_3\right)\right\} \\
    Y & \sim \text{Bernoulli}\left\{\text{expit}\left(0.2 + 3 X_1 - 2 X_1^2 + 2 X_2 + X_3 - 2 A\right)\right\}
\end{align*}

Our prediction model was a main effects logistic regression model estimated in the training data, i.e. $\mu\left(X^*\right) = \operatorname{expit}(\beta_0 + \beta_1 X_1 + \beta_2 X_2 + \beta_3 X_3)$. This model was misspecified with respect to the true data generating process. We assessed the counterfactual performance of the model in an untreated population using the AUC and the MSE, which for a binary outcome is equivalent to the Brier score \cite{brier_verification_1950}. In general, positing a parametric model for $h_0(X)=\mathrm{E}[(Y-\mu\left(X^*\right))^2 \mid X, A=0]$ may be difficult as the outcome is a squared difference. For binary outcomes, however, expanding the square shows that to estimate $h_0(X)$ it is enough to estimate $\operatorname{Pr}[Y=1 \mid X, A=0]$, which is the approach we used. To determine the effect of the specification of nuisance models $e_a(X)$ and $h_a(X)$ on performance estimates, we compared four estimators of AUC and MSE (${\psi}_{Na\ddot{i}ve}(0)$, ${\psi}_{IPW}(0)$, ${\psi}_{CL}(0)$, and ${\psi}_{DR}(0)$) using different combinations of correctly specified and misspecified models for $e_a(X)$ and $h_a(X)$:
\begin{enumerate}
    \item Correct $e_a(X)$ - main effects logistic regression model with linear and quadratic terms.
    \item Misspecified $e_a(X)$ - main effects logistic regression model with linear terms only terms.
    \item Correct $h_a(X)$ - main effects logistic regression model with linear and quadratic terms.
    \item Misspecified $h_a(X)$ - main effects logistic regression model with linear terms only terms.
\end{enumerate}
Finally, we also considered using more flexible estimation techniques for nuisance terms $e_a(X)$ and $h_a(X)$. Specifically, we estimate generalized additive models for both using the \texttt{mgcv} package in $\mathrm{R}$ entering all covariates as splines using the default options in the \texttt{gam} function.

\begin{table}[t]
    \centering
    \footnotesize
    \caption{Simulation results for experiment 2 comparing the performance of the na\"{i}ve, conditional loss (CL), inverse-probability weighting (IPW), and doubly robust (DR) estimators of the MSE and AUC under correct and misspecified nuisance models.}
 
\begin{threeparttable}
    \begin{tabular}{lcccccccc}
    \toprule
    \multicolumn{1}{c}{ } & \multicolumn{4}{c}{MSE} & \multicolumn{4}{c}{AUC} \\
    \cmidrule(l{3pt}r{3pt}){2-5} \cmidrule(l{3pt}r{3pt}){6-9}
    Estimator & Mean & $\sqrt{n}\times\text{SD}$ & $\sqrt{n}\times\text{Bias}$ & Percent & Mean & $\sqrt{n}\times\text{SD}$ & $\sqrt{n}\times\text{Bias}$ & Percent\\
    \midrule
    Na\"{i}ve & 0.207 & 0.176 & -0.140 & -2.1 & 0.742 & 0.491 & -1.335 & -5.4\\
    \addlinespace[0.3em]
    \multicolumn{9}{l}{Correct}\\
    \hspace{1em}CL & 0.212 & 0.333 & 0.015 & 0.2 & 0.783 & 0.767 & -0.045 & \vphantom{1} -0.2\\
    \hspace{1em}IPW & 0.212 & 0.517 & 0.011 & 0.2 & 0.782 & 1.258 & -0.062 & \vphantom{1} -0.3\\
    \hspace{1em}DR & 0.211 & 0.454 & 0.000 & 0.0 & 0.783 & 1.192 & -0.028 & -0.1\\
    \addlinespace[0.3em]
    \multicolumn{9}{l}{$e_a(X)$ misspecified}\\
    \hspace{1em}CL & 0.212 & 0.333 & 0.015 & 0.2 & 0.783 & 0.767 & -0.045 & -0.2\\
    \hspace{1em}IPW & 0.221 & 0.358 & 0.316 & 4.7 & 0.762 & 0.876 & -0.699 & -2.8\\
    \hspace{1em}DR & 0.212 & 0.349 & 0.016 & 0.2 & 0.782 & 0.841 & -0.066 & -0.3\\
    \addlinespace[0.3em]
    \multicolumn{9}{l}{$h_a(X)$ misspecified}\\
    \hspace{1em}CL & 0.217 & 0.356 & 0.194 & 2.9 & 0.777 & 0.803 & -0.224 & -0.9\\
    \hspace{1em}IPW & 0.212 & 0.517 & 0.011 & 0.2 & 0.782 & 1.258 & -0.062 & -0.3\\
    \hspace{1em}DR & 0.211 & 0.625 & 0.001 & 0.0 & 0.783 & 1.317 & -0.024 & -0.1\\
    \addlinespace[0.3em]
    \multicolumn{9}{l}{Both misspecified}\\
    \hspace{1em}CL gam & 0.213 & 0.348 & 0.052 & 0.8 & 0.782 & 0.800 & -0.063 & -0.3\\
    \hspace{1em}IPW gam & 0.214 & 0.422 & 0.075 & 1.1 & 0.778 & 1.032 & -0.181 & -0.7\\
    \hspace{1em}DR gam & 0.211 & 0.403 & 0.010 & 0.1 & 0.784 & 0.966 & -0.021 & -0.1\\
    Truth & 0.211 & &  &  & 0.784 & &  & \\
    \bottomrule
    \end{tabular}
    \begin{tablenotes}
    \item Average of estimates, estimated bias, estimated standard deviation (SD), and estimated relative bias for the na\"{i}ve empirical, weighting (IPW), conditional loss (CL), and doubly robust (DR) estimators. $n$ is the number of observations in the test set. Here, $h_a(X)$ is a model for $\operatorname{Pr}[Y=1 \mid X, A=a]$ and $e_a(X)$ denotes a model for $\Pr[A = a|X]$. Relative bias is calculated as $(\text{estimator} -\text{truth})/\text{truth}$. Correct and Misspecified refer to the nuisance models, $e_a(X)$ or $h_a(X)$ or both. In the final rows, gam indicates that a generalized additive model was used to estimate nuisance models. Results were averaged over 10,000 simulations.
    \end{tablenotes}
    \end{threeparttable}
\end{table}

Table 2 shows the results from experiment 2. As in the previous experiment, the na\"{i}ve empirical estimators of the AUC and MSE were biased relative to the true counterfactual values with a relative bias of $-2.1\%$ and $-5.4\%$ respectively. When all models were correctly specified, the weighting, conditional loss, and doubly robust estimators were all unbiased (absolute relative bias between 0.2\% to 0.3\%). When $\widehat e_a(X)$ was misspecified, the weighting estimator was biased (relative bias of 4.7\% and -2.8\%) but the conditional loss and doubly robust estimator were unbiased (absolute relative bias of 0.2\% to 0.3\%). Under misspecification of $h_a(X)$, the conditional loss estimator was biased (relative bias of 2.9\% and -0.9\%), but the weighting estimator and the doubly robust estimator were unbiased (absolute relative bias of 0.0\% to 0.3\%). When both models $e_a(X)$ and $h_a(X)$ were misspecified all estimators, including the doubly robust estimator, were biased. Finally, when a more flexible generalized additive model was used to estimate both $e_a(X)$ and $h_a(X)$, the doubly robust estimator was unbiased (absolute relative bias of 0.1\%). Across all scenarios, the weighting estimator generally had the largest standard errors and the conditional loss estimator had the smallest standard errors.

\section{Application to prediction of statin-na\"{i}ve risk} \label{sec:results}

We apply the proposed methods to evaluate the performance of two counterfactual prediction models targeting the statin-na\"{i}ve risk of cardiovascular disease: that is the risk in the same population if, contrary to fact, statins had been withheld. We compare one model that was explicitly tailored for the counterfactual estimand of interest and a second that was not. 

\subsection{Study design and data}
The Multi-Ethnic Study on Atherosclerosis (MESA) study is a population-based sample of 6,814 men and women aged 45 to 84 drawn from six communities (Baltimore; Chicago; Forsyth County, North Carolina; Los Angeles; New York; and St. Paul, Minnesota) in the United States between 2000 and 2002. The design, sampling procedures, and collection methods of the study have been described previously \cite{bild_multi-ethnic_2002}. Study teams conducted five examination visits between 2000 and 2011 in 18 to 24 month intervals focused on the prevalence, correlates, and progression of subclinical cardiovascular disease. These examinations included assessments of lipid-lowering medication use (primarily statins), as well as assessments of cardiovascular risk factors such as systolic blood pressure, serum cholesterol, cigarette smoking, height, weight, and diabetes. 

In a previous analysis, we used data from the MESA study to emulate a trial comparing continuous statin use versus no statins and benchmarked our results against those from published randomized trials. To construct a model of the statin-na\"{i}ve risk, we then emulated a single arm trial in which no one started statins over a 10-year follow up period. To determine trial eligibility, we followed the AHA guidelines \cite{grundy_scott_m_2018_2019} on statin use which stipulate that patients aged 40 to 75 with serum LDL cholesterol levels between 70 mg/dL and 190 mg/dL and no history of cardiovascular disease should initiate statins if their (statin-na\"{i}ve) risk exceeds 7.5\%. Therefore, we considered MESA participants who completed the baseline examination, had no previous history of statin use, no history of cardiovascular disease, and who met the criteria described in the guidelines (excluding the risk threshold) as eligible to participate in the trial. The primary endpoint was time to atherosclerotic cardiovascular disease (ASCVD), defined as nonfatal myocardial infarction, coronary heart disease death, or ischemic stroke. 

Follow up began at the second examination cycle to enable a ``wash out'' period for statin use and to ensure adequate pre-treatment covariates to control confounding (some were taken from the first cycle and others from the second). In the original analysis, we constructed a sequence of nested trials starting at each exam, however here for simplicity we limited our attention to the first trial. We used the questionnaire in examinations three through five to determine statin initiation over the follow up period. Because the exact timing of statin initiation was not known with precision, we estimated it by drawing a random month between the current and previous examinations (sensitivity analyses conducted in original study).

Of the 6,814 MESA participants who completed the baseline examination, 4,149 met the eligibility criteria for our trial emulation. There were 288 ASCVD events and 190 non-ASCVD deaths. For simplicity, here we dropped those lost to follow up and who first have competing events. For model training and evaluation, we further randomly split the dataset into training and test sets of equal size. 

\subsection{Model estimation and performance}

We compared two prediction models: one that was explicitly tailored to the statin-na\"{i}ve risk and a second that was not. Both models used the same regression specification with main effects of baseline predictors commonly used in cardiovascular risk prediction: age, sex, smoking status, diabetes history, systolic blood pressure, anti-hypertensive medication use and total and HDL serum cholesterol levels.

We tailored the first model for the statin-na\"{i}ve risk using inverse probability of censoring weights. In the emulated single arm trial, statin initiation can be viewed as ``non-adherence'' which can be adjusted for by inverse probability weighting, therefore we censored participants when they initiated statins. To calculate the stablized weights, we estimated two logistic regression models: one for the probability of remaining untreated given past covariate history (denominator model) and one for probability of remaining untreated given the selected baseline predictors (numerator model). The list of covariates in the weight models are given in section \ref{sec:covs} of the Appendix. To create a prediction model for the statin-na\"{i}ve risk, we used the estimated weights to fit a weighted logistic regression model conditional on the baseline predictors of interest. 

For comparison, we estimate a second traditional (factual) prediction model by regressing the observed ASCVD event indicator on the same set of baseline predictors, but ignoring treatment initiation over the follow up period. This approach targets the ``natural course'' risk (i.e., the risk under the statin initiation policies that prevailed at the time of the study) rather than the statin-na\"{i}ve risk. We estimate the model using standard logistic regression based on maximum likelihood.

To assess the performance of the models, we estimated the na\"{i}ve and counterfactual MSE in the test set. For the latter we used the conditional loss, inverse probability weighting, and doubly robust estimators of the MSE. Models for the initiation of treatment $e_a(X)$ and for the conditional loss $h_a(X)$ were implemented as main effects logistic regression models. As in the simulation example, to estimate the conditional loss it is sufficient to model the probability of the outcome alone. To quantify uncertainty, we used the non-parametric bootstrap with 1000 bootstrap replicates.

\begin{table}[t]
    \centering
    \caption{Estimated MSE in a statin-na\"{i}ve population for two prediction models using emulated trial data from MESA.}
    \begin{threeparttable}
        \begin{tabular}{lcccc}
        \toprule
        \multicolumn{1}{c}{ } & \multicolumn{2}{c}{MSE} & \multicolumn{2}{c}{AUC} \\
        \cmidrule(l{3pt}r{3pt}){2-3} \cmidrule(l{3pt}r{3pt}){4-5}
        Estimator & Logit & Weighted Logit & Logit & Weighted Logit\\
        \midrule
        Na\"{i}ve & 0.069 & 0.072 & 0.710 & 0.708\\
         & (0.003) & (0.003) & (0.013) & (0.014)\\
        CL & 0.086 & 0.085 & 0.719 & 0.727\\
         & (0.005) & (0.004) & (0.015) & (0.015)\\
        IPW & 0.109 & 0.099 & 0.753 & 0.778\\
         & (0.013) & (0.009) & (0.025) & (0.029)\\
        DR & 0.090 & 0.087 & 0.740 & 0.751\\
         & (0.006) & (0.005) & (0.023) & (0.023)\\
        \bottomrule
        \end{tabular}
        \centering
        \begin{tablenotes}[flushleft]
        \item The columns refer to the posited prediction model: Logit is an (unweighted) logistic regression model and weighted logit is a logistic regression model with inverse probability weights for remaining statin-free. The rows are the model performance estimates of the MSE and AUC. Na\"{i}ve is the empirical estimator using factual outcomes ($\widehat{\psi}_{Na\ddot{i}ve}(0)$), CL is the conditional loss estimator ($\widehat{\psi}_{CL}(0)$), IPW is the inverse probability weighting estimator ($\widehat{\psi}_{IPW}(0)$), DR is the doubly-robust estimator ($\widehat{\psi}_{DR}(0)$). Standard error estimates are shown in parentheses obtained via 1000 bootstrap replicates.
        \end{tablenotes}
        \end{threeparttable}
\end{table}

\subsection{Results}

Table 3 shows estimates of the AUC and MSE and the associated standard errors in a hypothetical statin-na\"{i}ve population for both prediction models using the na\"{i}ve empirical, conditional loss, weighting, and doubly robust estimators. Across both measures, the conditional loss, weighting, and doubly robust estimators yielded estimates that were worse (30-50\% greater MSE, 3-10\% lower AUC) than those of the na\"{i}ve empirical estimator, suggesting performance of both models in statin-na\"{i}ve population is worse than in the source population. Of the three estimators of the statin-na\"{i}ve performance, the weighting estimator had greater standard errors than the doubly robust estimator (by 10-100\%) as well as the conditional loss estimator (by 60-160\%). Consistent with the first simulation experiment, the inverse probability weighted logistic model, which was tailored to target the statin-na\"{i}ve risk, performed worse in the source population, but had lower MSE and higher AUC values in the counterfactual statin-na\"{i}ve population. There were sizeable differences (2-3 standard errors) in estimates of the MSE and AUC in the counterfactual statin-na\"{i}ve population across the proposed CL, IPW, and DR estimators. All three should give equivalent results in expectation if models for the nuisance functions $h_a(X^*)$ and $e_a(X^*)$ are correctly specified. Their divergence suggests at least one (but possibly both) of the nuisance functions may be misspecified. Finally, drawing on the results in section \ref{sec:calib} in the appendix, we estimate the counterfactual calibration of both prediction models in a statin-na\"{i}ve population using a weighted loess estimator where the weights are based on inverse probability of remaining statin free. Figure \ref{fig:calib} shows the results. As expected, the weighted prediction model, which was tailored to target the statin-na\"{i}ve risk, was better calibrated than the unweighted model, which generally underestimated the counterfactual risk.

\begin{figure}[p]
    \centering
    \includegraphics{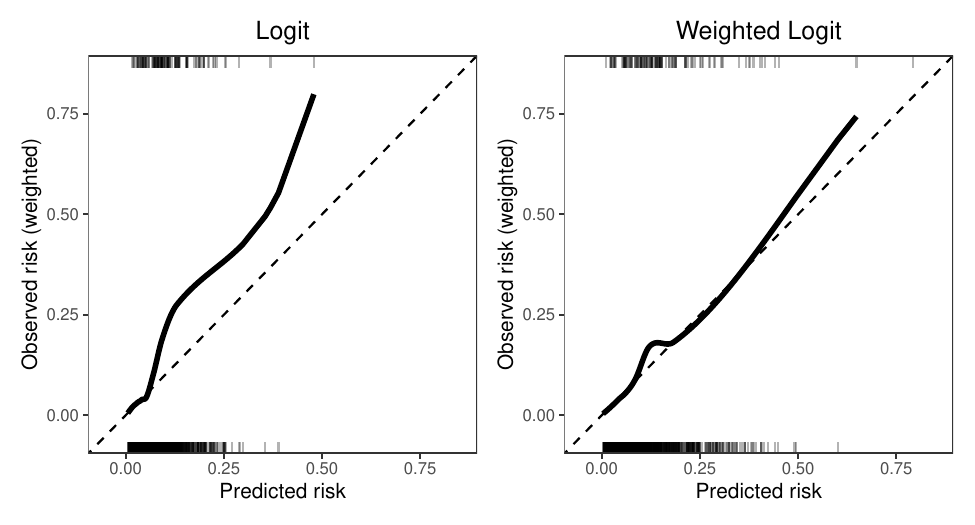}
    \caption{Risk calibration curves for counterfactual prediction models estimated using logistic regression with and with out inverse probability weights for statin initiation. The rug plot shows distribution of risk predictions among those who develop ASCVD (top) and those who don't (bottom). The black curve is a local-regression smoothed and inverse probability weighted estimate of the calibration curve targeting the statin-na\"{i}ve risk.\label{fig:calib}}
\end{figure}

\section{Discussion} \label{sec:discussion}

Many practical prediction problems call for the development and evaluation of counterfactual prediction models, for example, when the treatment distribution changes between the time training data are obtained and when the model will be deployed, or when predictions are meant to inform decisions about treatment initiation. In this paper, we considered the setting when the counterfactual prediction model has to be developed using observational training data. We described how to estimate counterfactual prediction models and the conditions necessary to identify these models. Separately, we also discussed how to adjust common measures of model performance to estimate the performance of counterfactual prediction model. Importantly, our results for performance measures are valid even when the prediction model is misspecified. A key insight was that performance measures that have not been tailored for counterfactual prediction will be biased. We also showed that performance under a hypothetical intervention can be assessed even if the prediction model itself is misspecified. We proposed estimators for these performance measures based on modeling the conditional loss, the probability of treatment, and a doubly robust estimator that can be used with data-adaptive estimators of either nuisance function. 

We focused on measures of performance under a particular treatment policy. However, prediction models may instead target the estimation of conditional treatment \textit{effects}, that is, comparisons between treatment regimes such as $\tau(X^*) = \E[Y^1 - Y^0 \mid X^*]$. In some cases, effects may be easier to communicate to end users or may be desirable to evaluate benefits versus harms of treatment initiation \cite{kent_predictive_2020}. However, absolute means and risks are common outputs of existing prediction models. Several authors have proposed model performance measures for conditional average treatment effects which are identifiable under similar assumptions to our own \cite{schuler_comparison_2018, rolling2014model, xu_calibration_2022, van2003unified, alaa_validating_2019}. From an estimation standpoint, methods for targeting the conditional average treatment effect and their performance have to balance the estimation of the conditional risk function for the outcome under different treatment levels with the estimation of the treatment effect function, with optimality depending on the relative smoothness of these functions \cite{kennedy_towards_2022}.

Throughout, we did not assume that the covariates needed to satisfy the exchangeability assumption were the same covariates used in the prediction model. This is an important aspect of our work because in practice predictors are often chosen based on their availability in a clinical setting rather than what would be optimal from a causal (or even predictive) perspective \cite{steyerberg_clinical_2019, coston_counterfactual_2021}. Moreover, our methods reflect the fact that confounding is a problem that needs to be adjusted for in the setting where the model is developed, whereas predictions often need to be optimized for the setting where the model will be applied (although as discussed below allowing the covariate distribution to differ between settings can introduce further subtleties).

One limitation of our approach is that it requires that a set of covariates sufficient to ensure exchanageability can be identified at the time of estimating the model. Violations of this exchangeability condition can be examined in sensitivity analyses \cite{robins_sensitivity_2000, steingrimsson_sensitivity_2023}, for example, to explore how violations of this assumption might affect estimates of model performance. Further work may also examine identification of counterfactual prediction models and their performance under alternative identifiability conditions.

In this work, we assumed that the distribution of predictors is the same in the training and deployment setting. In many cases, however, the covariate distributions are also likely to differ between settings \cite{bickel_discriminative_2009, sugiyama_covariate_2007}. Like differences in treatment initiation, differences in covariate distributions may cause the performance of the prediction model to degrade, particularly when it is misspecified. Methods for transporting prediction models from source to target populations which relate to our own have previously been proposed \cite{sugiyama_covariate_2007,bickel_discriminative_2007,sugiyama2012machine, steingrimsson_transporting_2023, li_estimating_2022, morrison_robust_2022} as have methods for transporting conditional average effects from trials to target populations with different covariate distributions \cite{mehrotra_transporting_2021, seamans_generalizability_2021, robertson_estimating_2021, robertson_regression-based_2023}. In future work, our results could be extended to allow for both differences in the distribution of treatment $f(A|X)$ and the covariates $f(X)$ between training and deployment settings to occur at the same time.

\section{Funding statement}
This work was supported by Patient-Centered Outcomes Research Institute (PCORI) award ME-2021C2-22365, National Library of Medicine (NLM) award R01LM013616, and National Heart, Lung, and Blood Institute (NHLBI) award R01HL136708 and training grant T32HL98048-10. The content is solely the responsibility of the authors and does not necessarily represent the official views of PCORI, PCORI's Board of Governors, PCORI's Methodology Committee, NLM, or NHLBI.

\clearpage

\printbibliography

\clearpage

\begin{appendix}
    \renewcommand{\thefigure}{A\arabic{figure}}
    \setcounter{figure}{0}
    
    \renewcommand{\thetable}{A\arabic{table}}
    \setcounter{table}{0}
    
    \renewcommand{\theequation}{A\arabic{equation}}
    \setcounter{equation}{0}

    \renewcommand{\thesection}{\Alph{section}}

    \newpage

\section{Time-fixed treatments} \label{sec:appendixa}

    \subsection{Identification of counterfactual prediction model estimands}\label{sec:model_proof}
    Our goal is to build a model that targets the expected potential outcome under a hypothetical intervention, e.g. the  model 
    $$\E[Y^a \mid X^*] = \mu_{a}(X^*).$$
    However, we do not observe $Y^a$ for all individuals. Nonetheless, as the following theorem shows, targets like $\E[Y^a \mid X^*]$ are identifiable from the observed data $(X, A, Y)$ under certain assumptions. 
    
    \begin{theorem}
     Under conditions A1-A3 in section \ref{sec:identifiability}, $\E[Y^a \mid X^*]$ is identified by the observed data functionals
    \begin{equation}
        \E[Y^a \mid X^*] = \E[\E[Y \mid X, A = a, D_{train} = 1] \mid X^*, D_{train} = 1]
    \end{equation}
    and
    \begin{equation}
        \E[Y^a \mid X^*] = \E\left[\frac{I(A = a)}{\Pr[A = a \mid X, D_{train} = 1]} Y \mid X^*, D_{train} = 1\right]
    \end{equation}
    in which case we can build a model for $\E[Y^a \mid X^*]$ by targeting either estimand in the training dataset.
    \end{theorem}
    
    \begin{proof}
        For the first representation we have 
        \begin{align*}
            \E[Y^a \mid X^*] & = \E[Y^a \mid X^*, D_{train} = 1] \\
            & = \E(\E[Y^a\mid X, D_{train} = 1] \mid X^*, D_{train} = 1) \\
            & = \E(\E[Y^a\mid X, A = a, D_{train} = 1] \mid X^*, D_{train} = 1) \\
            & = \E(\E[Y \mid X, A = a, D_{train} = 1] \mid X^*, D_{train} = 1) 
        \end{align*}
        where the first line follows from the random sampling of the training set, the second from the law of iterated expectations, the third from the exchangeability condition A1, and the fourth from the consistency condition A2. Recall that $X^*$ is a subset of $X$. For the second representation, we show that it is equivalent to the first 
        \begin{align*}
            \E[Y^a \mid X^*] &= \E(\E[Y \mid X, A = a, D_{train} = 1] \mid X^*, D_{train} = 1)  \\
            &= \E\left(\E\left[\frac{I(A = a)}{\Pr[A = a \mid X,D_{train} = 1]}Y \mid X,D_{train} = 1\right] \mid X^*, D_{train} = 1\right)\\
            &= \E\left[\frac{I(A = a)}{\Pr[A = a \mid X,D_{train} = 1]}Y \mid X^*, D_{train} = 1\right]
        \end{align*}
        where the second line follows from the definition of conditional expectation, and the last reverses the law of iterated expectations.
    \end{proof}

    \subsection{Identification of general loss functions}\label{sec:tf_proof}
    Many common model performance measures, such as the mean squared error, Brier score, and absolute error are special cases of a generic loss function $L(Y, \widehat{\mu}(X^*))$. To assess the performance of counterfactual predictions, we would like to estimate 
    $$\psi(a) = \E[L(Y^a, \widehat{\mu}(X^*))]$$
    where $Y^a$ is not observed for all individuals. The following theorem states that, under the conditions of section \ref{sec:identifiability}, $\psi(a)$ is identified using the observed data alone. 
    \begin{theorem}
        Under conditions A1-A3 in section \ref{sec:identifiability} and the time-fixed setup described in section \ref{sec:setup}, the expected loss is identified by the functionals
        \begin{equation}\label{eqn:app_cl_estimand}
            \psi(a) = \E\left(\E[L\{Y, \widehat{\mu}(X^*)\} \mid X, A=a, D_{test} = 1] \mid D_{test} = 1\right)
        \end{equation}
        and 
        \begin{equation}\label{eqn:app_ipw_estimand}
            \psi(a) = \E\left[\frac{I(A = a)}{\Pr[A = a \mid X, D_{test} = 1]}L\{Y, \widehat{\mu}(X^*)\} \mid D_{test} = 1\right]
        \end{equation}
        in the test set for generic counterfactual loss function $L(Y^{a}, \widehat{\mu}(X^*))$.
    \end{theorem}
    
    \begin{proof}
        For the first representation we have 
        \begin{align*}
            \psi(a) &= \E[L\{Y^{a}, \widehat{\mu}(X^*)\}] \\
            & = \E[L\{Y^{a}, \widehat{\mu}(X^*)\}\mid D_{test} = 1] \\
            & = \E(\E[L\{Y^{a}, \widehat{\mu}(X^*)\}\mid X, D_{test} = 1] \mid D_{test} = 1) \\
            & = \E(\E[L\{Y^{a}, \widehat{\mu}(X^*)\}\mid X, A = a, D_{test} = 1] \mid D_{test} = 1) \\
            & = \E(\E[L\{Y, \widehat{\mu}(X^*)\}\mid X, A = a, D_{test} = 1] \mid D_{test} = 1) 
        \end{align*}
        where the first line follows from the definition of $\psi(a)$, the second from random sampling of the test set, the third from the law of iterated expectations, the fourth from the exchangeability condition A1, and the fifth from the consistency condition A2. Recall that $X^*$ is a subset of $X$. For the second representation, we show that it is equivalent to the first 
        \begin{align*}
            \psi(a) &= \E(\E[L\{Y, \widehat{\mu}(X^*)\}\mid X,A = a, D_{test} = 1] \mid D_{test} = 1) \\
            &= \E\left(\E\left[\frac{I(A = a)}{\Pr[A = a \mid X,D_{test} = 1]}L\{Y, \widehat{\mu}(X^*)\} \mid X,D_{test} = 1\right] \mid D_{test} = 1\right)\\
            &= \E\left[\frac{I(A = a)}{\Pr[A = a \mid X,D_{test} = 1]}L\{Y, \widehat{\mu}(X^*)\} \mid D_{test} = 1\right]
        \end{align*}
        where the second line follows from the definition of conditional expectation, and the last reverses the law of iterated expectations.
    \end{proof}
    
    \subsection{Plug-in estimation for general loss functions}
    Using sample analogs for the identified expressions \ref{eqn:app_cl_estimand} and \ref{eqn:app_ipw_estimand}, we obtain two plug-in estimators for the expected loss for a generalized loss function
    \begin{equation*}
        \widehat{\psi}_{CL}(a) = \frac{1}{n_{test}}\sum_{i=1}^nI(D_{test, i} = 1)\widehat{h}_a(X_i)
    \end{equation*}
    and 
    \begin{equation*}
        \widehat{\psi}_{IPW}(a) = \frac{1}{n_{test}} \sum_{i=1}^n \frac{I(A_i = a, D_{test, i} = 1)}{ \widehat{e}_a(X_i)} L\{Y_i, \widehat{\mu}(X^*_i)\}
    \end{equation*}
    where $\widehat{h}_a(X)$ is an estimator for $\E[L\{Y, \widehat{\mu}(X^*)\}\mid X,A = a, D_{test} = 1]$ and $\widehat{e}_a(X)$ is an estimator for $\Pr[A = a \mid X,D_{test} = 1]$. Using the terminology in \textcite{morrison_robust_2022}, we call the first plug-in estimator the conditional loss estimator $ \widehat{\psi}_{CL}(a)$ and the second the inverse probability weighted estimator $\widehat{\psi}_{IPW}(a)$. 

    
    \subsection{Random and dynamic regimes}\label{sec:randomdynamic}
    Above we consider static interventions which set treatment $A$ to a particular value $a$. We might also consider interventions which probabilistically set $A$ based on a known density, possibly conditional on pre-treatment covariates, e.g. $f^*(A \mid X)$ where $f^*$ denotes an investigator specified density to contrast it with the observed density $f(A\mid X)$. For instance, instead of a counterfactual prediction if everyone or no one had been treated, we may be interested in the prediction if 20\% or 50\% were treated. We term such an intervention a \textit{random} or \textit{stochastic} intervention to contrast it with \textit{deterministic} interventions considered previously. Random interventions are closer to the counterfactual interventions of interest under dataset shift which may be approximated as probabilistic changes in treatment initiation due to changes in guidelines or prescribing patterns or the wider-availability. Denote by $g$ the intervention which assigns $A$ according to $f^*(A \mid X)$. For general loss functions, the expected loss under $g$, i.e. $L\{Y^{g}, \widehat{\mu}(X^*)\}$, is identified by the functionals
    \begin{equation}\label{eqn:rand_cl_estimand}
        \psi(g) = \E_{f^*(A|X)}\left(\E[L\{Y, \widehat{\mu}(X^*)\} \mid X, A, D_{test} = 1] \mid D_{test} = 1\right)
    \end{equation}
    and 
    \begin{equation}\label{eqn:rand_ipw_estimand}
        \psi(g) = \E\left[\frac{f^*(A|X)}{f(A \mid X, D_{test} = 1)}L\{Y, \widehat{\mu}(X^*)\} \mid D_{test} = 1\right]
    \end{equation}
    in the test set under Assumptions A1 to A3. The proof is the same given for static interventions above by recognizing that 
    \begin{align*}
        \E_{f^*(A|X)}&\left(\E[L\{Y^g, \widehat{\mu}(X^*)\} \mid X, A, D_{test} = 1] \mid D_{test} = 1\right) = \\
        &\qquad \sum_a\E[L\{Y^a, \widehat{\mu}(X^*)\} \mid X, A, D_{test} = 1]\operatorname{Pr}^*[A=a|X]
    \end{align*}
    where $\operatorname{Pr}^*[A=a|X]$ is the probability of receiving treatment $a \in \{0,1\}$ under the treatment assignment strategy $f^*(A|X)$.

    For plugin estimation for a binary treatment, the corresponding sample analogs are given by
    \begin{align}
        \widehat\psi_{CL}(g) = \frac{1}{n_{test}} \sum_{i=1}^n I(D_{test,i}=1) \left( \operatorname{Pr}^*[A=1|X_i] \widehat  h_{1}(X_i) + \operatorname{Pr}^*[A=0|X_i] \widehat  h_{0}(X_i) \right)
    \end{align}
    and 
    \begin{align}
        \begin{split}
        \widehat\psi_{IPW}(g) &= \frac{1}{n_{test}} \sum_{i=1}^n I(D_{test,i}=1)  L\{Y_i, \widehat{\mu}(X^*_i)\} \\
        & \qquad \qquad \times \left( I(A_i = 1)\frac{\operatorname{Pr}^*[A=1|X_i]}{ \widehat{e}_1(X_i)} + I(A_i = 0) \frac{\operatorname{Pr}^*[A=0|X_i]}{ \widehat{e}_0(X_i)} \right)
        \end{split}
    \end{align}
     where $\widehat{h}_a(X)$ is an estimator for $\E[L\{Y, \widehat{\mu}(X^*)\}\mid X,A = a, D_{test} = 1]$ and $\widehat{e}_a(X)$ is an estimator for $\Pr[A = a \mid X,D_{test} = 1]$. We use the same terminology as previous.

\newpage
\section{Time-varying treatments}\label{sec:timevarying}
\subsection{Set up}
Here we extend the set up of section \ref{sec:setup} in the case that treatment is time-varying over the follow up period. We now observe $n$ i.i.d. longitudinal samples $\{O_i\}_{i=1}^n$ from a source population. For each observation, let 
\[O =(\overline{X}_K, \overline{A}_K, Y_{K+1})\]
where $X_k$ is a vector of covariates, $A_k$ is an indicator of treatment, and $Y_k$ is an event indicator all measured at time $k$, where $k \in \{0,1,\ldots, K\}$. Overbars denote the full history of a variable, such that $\overline{X}_k = (X_0,\dots, X_k)$. Again we assume the goal is to build a prediction model for end of follow up outcome $Y_{K+1}$ conditional on baseline covariates $X^*$ which are now a subset of baseline covariates $X_0$, i.e. $X^* \subset X_0$. An example DAG for a two time point process is shown in Figure \ref{fig:dag2}.

We would like to assess the performance of the model in a counterfactual version of the source population in which a new treatment policy is implemented. As previously, $Y^a$ is the potential outcome under an intervention which sets treatment $A$ to $a$. For a sequence of time-varying treatments $\overline{A}_k$, we further define a \textit{treatment regime} as a collection of functions $\{g_k(\overline{a}_{k-1}, \overline{x}_k): k=0,\ldots, K\}$ for determining treatment assignment at each time $k$, possibly based on past treatment and covariate history. For a hypothetical treatment regime $g$, we would like to determine the performance of fitted model $\widehat{\mu}(X^*)$ under the new regime by estimating the expected loss 
$$\psi(g) = \E[L\{Y^g, \widehat{\mu}(X^*)\}]$$
for generalized loss function $L\{Y^g, \widehat{\mu}(X^*)\}$.

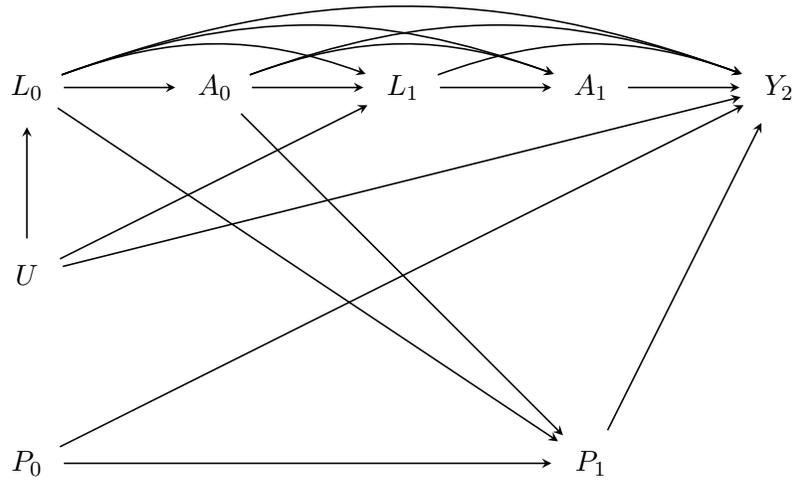
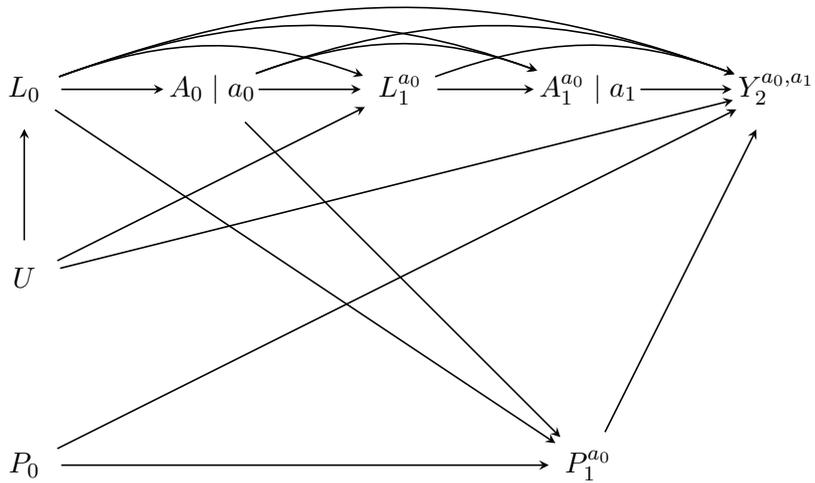
\begin{figure}[p]
    \centering
    \begin{subfigure}{\linewidth}
        \centering
        \begin{tikzpicture}[> = stealth, shorten > = 1pt, auto, node distance = 2.5cm, inner sep = 0pt,minimum size = 0.5pt, very thick]
            \tikzstyle{every state}=[
              draw = white,
              fill = white
            ]
            \node[state] (l0) {$L_{0}$};
            \node[state] (a0) [right of=l0] {$A_{0}$};
            \node[state] (l1) [right of=a0] {$L_1$};
            \node[state] (a1) [right of=l1] {$A_1$};
            \node[state] (y1) [right of=a1] {$Y_2$};
            \node[state] (u0) [below of=l0] {$U$};
            \node[state] (p0) [below of=u0] {$P_{0}$};
            \node[state] (b0) [right of=p0] {};
            \node[state] (b1) [right of=b0] {};
            \node[state] (p1) [right of=b1] {$P_{1}$};
            
            \path[->] (l0) edge node {} (a0);
            \path[->] (l0) edge node {} (p1);
            \path[->] (l0) edge [out=20, in=160, looseness=1] node {} (l1);
            \path[->] (l0) edge [out=20, in=160, looseness=1] node {} (a1);
            \path[->] (l0) edge [out=20, in=160, looseness=1] node {} (y1);
        
            \path[->] (a0) edge node {} (p1);
            \path[->] (a0) edge node {} (l1);
            \path[->] (a0) edge [out=20, in=160, looseness=1]  node {} (a1);
            \path[->] (a0) edge [out=20, in=160, looseness=1]  node {} (y1);
            
            \path[->] (p0) edge node {} (p1);
            \path[->] (p0) edge node {} (y1);
            
            \path[->] (p1) edge node {} (y1);
                
            \path[->] (l1) edge node {} (a1);
            \path[->] (l1) edge [out=20, in=160, looseness=1]  node {} (y1);
            
            \path[->] (a1) edge node {} (y1);
        
            \path[->] (u0) edge node {} (y1);
            \path[->] (u0) edge node {} (l0);
            \path[->] (u0) edge node {} (l1);
            \end{tikzpicture}
            \caption{Example two time point directed acyclic graph for prediction.}
    \end{subfigure}
    \vspace{2em}\\
    \begin{subfigure}{\linewidth}
        \centering
        \begin{tikzpicture}[> = stealth, shorten > = 1pt, auto, node distance = 2.5cm, inner sep = 0pt,minimum size = 0.5pt, very thick]
            \tikzstyle{every state}=[
              draw = white,
              fill = white
            ]
            \node[state] (l0) {$L_{0}$};
            \node[state] (a0) [right of=l0] {$A_{0} \mid a_0$};
            \node[state] (l1) [right of=a0] {$L^{a_0}_1$};
            \node[state] (a1) [right of=l1] {$A^{a_0}_1 \mid a_1$};
            \node[state] (y1) [right of=a1] {$Y^{a_0, a_1}_2$};
            \node[state] (u0) [below of=l0] {$U$};
            \node[state] (p0) [below of=u0] {$P_{0}$};
            \node[state] (b0) [right of=p0] {};
            \node[state] (b1) [right of=b0] {};
            \node[state] (p1) [right of=b1] {$P^{a_0}_{1}$};
            
            \path[->] (l0) edge node {} (a0);
            \path[->] (l0) edge node {} (p1);
            \path[->] (l0) edge [out=20, in=160, looseness=1] node {} (l1);
            \path[->] (l0) edge [out=20, in=160, looseness=1] node {} (a1);
            \path[->] (l0) edge [out=20, in=160, looseness=1] node {} (y1);
        
            \path[->] (a0) edge node {} (p1);
            \path[->] (a0) edge node {} (l1);
            \path[->] (a0) edge [out=20, in=160, looseness=1]  node {} (a1);
            \path[->] (a0) edge [out=20, in=160, looseness=1]  node {} (y1);
            
            \path[->] (p0) edge node {} (p1);
            \path[->] (p0) edge node {} (y1);
            
            \path[->] (p1) edge node {} (y1);
                
            \path[->] (l1) edge node {} (a1);
            \path[->] (l1) edge [out=20, in=160, looseness=1]  node {} (y1);
            
            \path[->] (a1) edge node {} (y1);
        
            \path[->] (u0) edge node {} (y1);
            \path[->] (u0) edge node {} (l0);
            \path[->] (u0) edge node {} (l1);
            \end{tikzpicture}
            \caption{Single world intervention graph of intervention on $A_0$ and $A_1$.}
    \end{subfigure}
    \caption{Example directed acyclic graph (DAG) and single world intervention graph (SWIG) for a two time point process.}
    \label{fig:dag2}
\end{figure}

\subsection{Identifiability conditions}
We now consider modified identifiability conditions under time-varying treatment initiation. For all $k$ from 0 to $K$, we require
\begin{enumerate}
    \item[B1.] \textit{Exchangeability:} $Y^g_{K+1} \perp \!\!\! \perp A_k \mid \overline{X}_k, \overline{A}_{k-1}$
    \item[B2.] \textit{Consistency:} $Y_{K+1} = Y^g_{K+1}$\text{ and } $\overline{X}_{k} = \overline{X}^g_{k}$ \text{ if } $\overline{A}_k = \overline{a}_k^g$
    \item[B3.] \textit{Positivity:} $1 > \Pr[A_k = a_k \mid \overline{X}_k = \overline{x}_k, \overline{A}_{k -1} = \overline{a}_{k-1}] > 0$
\end{enumerate}

\subsection{Identification of general loss functions}
The following theorem extends Theorem 2 to the case of a time-varying treatment. 
\begin{theorem}
    Under conditions B1-B3 above, the expected counterfactual loss under time-varying regime $g$ is identified by the functionals
    \begin{align}\label{eqn:cl_tv_estimand}
    \begin{split}
        \psi(g) &= \E_{X_0}\bigg[\E_{X_1}\bigg\{\ldots \E_{X_{K-1}}\bigg(\E_{X_{K}}[L\{Y, \widehat{\mu}(X^*)\} \mid \overline{X}_K, \overline{A}_K=\overline{a}^g_K, D_{test} = 1] \\
        & \qquad \big\vert\; \overline{X}_{K-1}, \overline{A}_{K-1}=\overline{a}^g_{K-1}, D_{test} = 1\bigg) \ldots \;\big\vert\; X_{0}, A_{0}=a^g_{0}, D_{test} = 1\bigg\}\;\big\vert\; D_{test} = 1\bigg]
    \end{split}
    \end{align}
and 
    \begin{equation}\label{eqn:ipw_tv_estimand}
        \psi(g) = \E\left[\frac{I(\overline{A}_K = \overline{a}^g_K, D_{test} = 1)}{\prod_{k=0}^K\Pr[A_k = a^g_k \mid \overline{X}_k, \overline{A}_{k-1} = \overline{a}^g_{k-1}, D_{test} = 1]}L\{Y, \widehat{\mu}(X^*)\} \mid D_{test} = 1\right]
    \end{equation}
in the test set for general loss function $L\{Y^{g}, \widehat{\mu}\}$, where the first expression is a sequence of iterated expectations and the second is an inverse-probability weighted expectation.
\end{theorem}

\begin{proof}
    For the first representation we have 
    \begin{align*}
        \psi(g) &= \E[L\{Y^{g}, \widehat{\mu}(X^*)\}] \\
        & = \E[L\{Y^{g}, \widehat{\mu}(X^*)\}\mid D_{test} = 1] \\
        & = \E(\E[L\{Y^{g}, \widehat{\mu}(X^*)\}\mid X_0, D_{test} = 1] \mid D_{test} = 1) \\
        & = \E(\E[L\{Y^{g}, \widehat{\mu}(X^*)\}\mid X_0, A_0 = a^g_0, D_{test} = 1] \mid D_{test} = 1) 
    \end{align*}
    where the first line follows from the definition of $\psi(g)$, the second from random sampling of the test set, the third from the law of iterated expectations, and the fourth from the exchangeability condition B1. Arguing recursively from $k = 0$ to $K$, we can repeatedly invoke iterated expectations and exchanageability to insert $\overline{X}_k$ and $\overline{A}_k = \overline{a}^g_k$, such that
    \begin{align*}
        \psi(g) &= \E_{X_0}\bigg[\E_{X_1}\bigg\{\ldots \E_{X_{K}}\bigg(E[L\{Y^g, \widehat{\mu}(X^*)\} \mid \overline{X}_K, \overline{A}_K=\overline{a}^g_K, D_{test} = 1] \\
        & \qquad \big\vert\; \overline{X}_{K}, \overline{A}_{K}=\overline{a}^g_{K}, D_{test} = 1\bigg) \ldots \;\big\vert\; X_{0}, A_{0}=a^g_{0}, D_{test} = 1\bigg\}\;\big\vert\; D_{test} = 1\bigg]\\
        &= \E_{X_0}\bigg[\E_{X_1}\bigg\{\ldots \E_{X_{K}}\bigg(\E[L\{Y, \widehat{\mu}(X^*)\} \mid \overline{X}_K, \overline{A}_K=\overline{a}^g_K, D_{test} = 1] \\
        & \qquad \big\vert\; \overline{X}_{K}, \overline{A}_{K}=\overline{a}^g_{K}, D_{test} = 1\bigg) \ldots \;\big\vert\; X_{0}, A_{0}=a^g_{0}, D_{test} = 1\bigg\}\;\big\vert\; D_{test} = 1\bigg]
    \end{align*}
    where the last line follows by consistency condition B2. For the second representation, note that for the inner most expectations we can proceed as previously
    \begin{align*}
        & \E(\E[L\{Y, \widehat{\mu}(X^*)\} \mid \overline{X}_K, \overline{A}_k=\overline{a}^g_K, D_{test} = 1] \mid \overline{X}_{K-1}, \overline{A}_{k-1}=\overline{a}^g_{K-1}, D_{test} = 1) \\
        &= \E\left(\E\left[W_K L\{Y, \widehat{\mu}(X^*)\} \mid \overline{X}_K, \overline{A}_{K-1}, D_{test} = 1\right] \mid \overline{X}_{K-1}, \overline{A}_{K-1}, D_{test} = 1\right)\\
        &= \E\left(W_K \E\left[L\{Y, \widehat{\mu}(X^*)\} \mid \overline{X}_K, \overline{A}_{K-1}, D_{test} = 1\right] \mid \overline{X}_{K-1}, \overline{A}_{K-1}, D_{test} = 1\right)\\
        &= \E\left[W_K L\{Y, \widehat{\mu}(X^*)\} \mid \overline{X}_{K-1}, \overline{A}_{K-1}, D_{test} = 1\right]
    \end{align*}
    where the second line follows from the definition of conditional expectation, the third removes the constant fraction outside expectation, and the last reverses the law of iterated expectations and where 
    $$W_K = \frac{I(A_K = a_K^g, D_{test} = 1)}{\Pr[A_K = a_K^g \mid \overline{X}_K, \overline{A}_{K-1}, D_{test} = 1]}$$
    Arguing recursively from $k = 0$ to $K$, we get
    \begin{align*}
        \psi(g) &= \E\left[\frac{I(\overline{A}_K = \overline{a}^g_K, D_{test} = 1)}{\prod_{k=0}^K\Pr[A_k = a^g_k \mid \overline{X}_k, \overline{A}_{k-1} = \overline{a}^g_{k-1}, D_{test} = 1]}L\{Y, \widehat{\mu}(X^*)\} \mid D_{test} = 1\right] 
    \end{align*}
    which is the inverse-probability weighted representation with weights equal to 
    $$W_k = \frac{I(\overline{A}_K = \overline{a}^g_K, D_{test} = 1)}{\prod_{k=0}^K\Pr[A_k = a^g_k \mid \overline{X}_k, \overline{A}_{k-1} = \overline{a}^g_{k-1}, D_{test} = 1]}.$$
\end{proof}

\subsection{Plug-in estimation}

Using sample analogs for the identified expressions \ref{eqn:cl_tv_estimand} and \ref{eqn:ipw_tv_estimand}, we obtain two plug-in estimators for the expected counterfactual loss under a generalized loss function
\begin{equation*}
    \widehat{\psi}_{CL}(g) = \frac{1}{n_{test}}\sum_{i=1}^nI(D_{test, i} = 1)\widehat{h}_{a_0}(X_i)
\end{equation*}
and 
\begin{equation*}
    \widehat{\psi}_{IPW}(g) = \frac{1}{n_{test}}\sum_{i=1}^n \frac{I(\overline{A}_{K,i} = \overline{a}^g_{K,i}, D_{test,i} = 1)}{\prod_{k=0}^K \widehat{e}_{a_k}(X_i)} L\{Y_i, \widehat{\mu}(X^*_i)\}
\end{equation*}
where $h_{a_{K+1}}(X) = L\{Y, \widehat{\mu}(X^*)\}$ and $h_{a_{t+1}}(X)$ is recursively defined for $t$ = $K, \ldots, 0$
\[h_{a_t}(X_t) : \E[ h_{a_{t+1}}(X_{t+1}) \mid \overline{X}_t, \overline{A}_t = \overline{a}^g_t, D_{test} = 1]\]
Similarly, $\widehat{e}_{a_k}(X)$ is an estimator for $\Pr[A_k = a^g_k \mid \overline{X}_k, \overline{A}_{k-1} = \overline{a}^g_{k-1}, D_{test} = 1]$. Note that as the number of time points (i.e. $K$) increases, the proportion in the test set who actually follow the regime of interest, i.e. those for whom $I(\overline{A}_K = \overline{a}^g_K, D_{test,i} = 1)=1$ may be prohibitively small, in which case plug-in estimation may not be feasible. In this case, additional modeling assumptions will be necessary to borrow information from other regimes.
\newpage

\section{Doubly robust estimators} \label{sec:dr}

\subsection{Efficient influence function}
We have shown previously that, under the identifiability conditions of section \ref{sec:identifiability}, the expected counterfactual loss of a generalized loss function $L\{Y^a, \mu(X^*)\}$ is identified by the observed data functional
\begin{equation*}\label{eqn:cl_estimand}
    \psi(a) = \E\left(\E[L\{Y, \mu(X^*)\} \mid X, A=a] \right).
\end{equation*}
In the following theorem, we identify the efficient influence function for $\psi(a)$ under a nonparametric model for the observed data in the test set. 

\begin{theorem}
    The influence function for $\psi(a)$ under a nonparametric model for the observable data $O = (X, A, Y)$ is 
\begin{align*}
    \chi_{P_0}^1 &= \frac{I(A = a)}{\Pr[A = a \mid X]}(L\{Y, \mu(X^*)\} - \E[L\{Y, \mu(X^*)\} \mid X, A=a])  \; + \\
    & \qquad (\E[L\{Y, \mu(X^*)\} \mid X, A=a] - \psi(a)).
\end{align*}
As the influence function under a nonparametric model is always unique, it is also the efficient influence function. 
\end{theorem}

\begin{proof}
To show that $\chi_{P_0}^1$ is the efficient influence function, we will use the well-known fact that the influence function is a solution to 
\begin{equation*}
    \frac{d}{dt} \psi_{P_t}(a)\bigg\vert_{t=0} = \E_{P_0}(\chi_{P_0}^1g_{P_0})
\end{equation*}
where $g_{P_0}$ is the score of the obeservable data under the true law $P_0$ and $P_t$ is a parametric submodel indexed by $t \in [0,1]$ and the pathwise derivative of the submodel is evaluated at $t = 0$ corresponding to the true law $P_0$. Let $h_a(X) = E_{P_0}[L\{Y, \mu(X^*)\} \mid X, A=a]$. Beginning with the left hand side
\begin{align*}
    \frac{d}{dt} \psi_{P_t}(a)\bigg\vert_{t=0} &=\frac{d}{dt} \E_{P_t}\left(\E_{P_t}[L\{Y, \mu(X^*)\} \mid X, A=a] \right)\bigg\vert_{t=0} \\
    &=\frac{\partial}{\partial t} \E_{P_t}\left(\E_{P_0}[L\{Y, \mu(X^*)\} \mid X, A=a] \right)\bigg\vert_{t=0} \;+ \\
    &\qquad  \E_{P_0}\left(\frac{\partial}{\partial t} \E_{P_t}[L\{Y, \mu(X^*)\} \mid X, A=a] \bigg\vert_{t=0} \right) \\
    &=\E_{P_0}\left[\left\{h_a(X) - \psi(a) \right\}g_{X, A, Y}(O)\right] \;+ \\
    &\qquad  \E_{P_0}\left\{\left( \frac{I(A = a)}{\Pr[A = a \mid X]} \bigg[L\{Y, \mu(X^*)\} - h_a(X) \bigg]\right)g_{X, A, Y}(O)\right\} \\
    &= \E_{P_0}\left\{\left(h_a(X) - \psi(a) + \frac{I(A = a)}{\Pr[A = a \mid X]} \bigg[L\{Y, \mu(X^*)\} - h_a(X) \bigg]\right)g_{X, A, Y}(O)\right\} 
\end{align*}
where the first line is the definition, the second line applies the chain rule, the third applies definition of the score, and the last uses linearity of expectations. Returning to original supposition, it follows that the influence function is 
\begin{align*}
    \chi_{P_0}^1 &= \frac{I(A = a)}{\Pr[A = a \mid X]}(L\{Y, \mu(X^*)\} - \E[L\{Y, \mu(X^*)\} \mid X, A=a])  \; + \\
    & \qquad (\E[L\{Y, \mu(X^*)\} \mid X, A=a] - \psi(a)).
\end{align*}
\end{proof}

\subsection{One-step estimator}
Given the efficient influence function above and  random sampling in the test set, the one-step estimator for $\psi(a)$ is given by
\begin{equation*}
    \widehat{\psi}_{DR}(a) = \frac{1}{n_{test}}\sum_{i=1}^n I(D_{test, i} = 1)\widehat{h}_a(X_i) + \frac{I(A_i = a, D_{test, i} = 1)}{\widehat{e}_a(X_i)} \left[ L\{Y_i, \mu(X^*_i)\} - \widehat{h}_a(X_i)\right]
\end{equation*}

\subsection{Asymptotic properties}
In previous sections, the asymptotic properties of $\widehat{\psi}_{CL}(a)$ and $\widehat{\psi}_{IPW}(a)$ follow from standard parametric theory\footnote{after separating estimation of $\mu(X^*)$ from the evaluation of performance by random partition of test set.}. However, the asymptotic properties of $\widehat{\psi}_{DR}(a)$ are complicated by the estimation of two nuisance functions, $\widehat{h}_a(X)$ and $\widehat{e}_a(X)$, and the fact that, we do not immediately assume a parametric model for either. To simplify the derivation of the large sample properties of $\widehat{\psi}_{DR}(a)$ we begin by defining
$$
H\left(e_a^{\prime}(X), h_a^{\prime}(X)\right)=h_a^{\prime}(X)+\frac{I(A = a)}{e_a^{\prime}(X)}\left[L\left(Y, \mu\left(X^*\right)\right)-h_a^{\prime}(X)\right]
$$
for arbitrary functions $e_a^{\prime}(X)$, and $h_a^{\prime}(X)$. Here we suppress the dependence on being in the test set for ease of exposition, but note that the rest procedes the same if we were to limit our focus to the test set. Note, the doubly robust estimator can be written as $\widehat{\psi}_{DR}(a)=\frac{1}{n} \sum_{i=1}^n H\left(\widehat{e}_a\left(X_i\right), \widehat{h}_a\left(X_i\right)\right)$. We define the probability limits of $\widehat{e}_a(X)$ and $\widehat{h}_a(X)$ as $e_a^*(X)$ and $h_a^*(X)$, respectively. By definition, when $\widehat{e}_a(X)$ and $\widehat{h}_a(X)$ are correctly specified, the limits are  $e_a^*(X)=$ $\operatorname{Pr}[A=a \mid X]$ and $h^*_a(X)=\mathrm{E}\left[L\left(Y, \mu\left(X^*\right)\right) \mid X, A=a\right]$.

To derive the asymptotic properties of $\widehat{\psi}_{DR}(a)$, we make the following assumptions:

\begin{enumerate}
    \item[D1.] $H(\widehat{e}_a(X), \widehat{h}_a(X))$ and its limit $H\left(e^*_a(X), h^*_a(X)\right)$ fall in a Donsker class.
    \item[D2.]  $\left\|H(\widehat{e}_a(X), \widehat{h}_a(X))-H\left(e^*_a(X), h^*_a(X)\right)\right\| \stackrel{P}{\longrightarrow} 0$.
    \item[D3.] (Finite second moment). $\mathrm{E}\left[H\left(e^*_a(X), h^*_a(X)\right)^2\right]<\infty$.
    \item[D4.] (Model double robustness). At least one of the models $\widehat{e}_a(X)$ or $\widehat{h}_a(X)$ is correctly specified. That is, at least one of $e^*_a(X)=\operatorname{Pr}[A=a \mid X]$ or $h^*_a(X)=\mathrm{E}\left[L\left(Y, \mu\left(X^*\right)\right) \mid X, A=a\right]$ holds, but not necessarily both.
\end{enumerate}

Assumption D1 is a well-known restriction on the complexity of the functionals $\widehat{e}_a(X)$ and $\widehat{h}_a(X)$. As long as $\widehat{e}_a(X), \widehat{h}_a(X), e^*_a(X)$, and $h^*_a(X)$ are Donsker and all are uniformly bounded then Assumption D1 holds by the Donsker preservation theorem. Many commonly used models such as generalized linear models fall within the Donsker class. This requirement can be further relaxed through sample-splitting, in which case more flexible machine learning algorithms such as random forests, gradient boosting, or neural networks may be used to estimate $\widehat{e}_a(X)$ and $\widehat{h}_a(X)$. 

Using Assumptions D1 through D4, below we prove:
\begin{enumerate}
    \item (Consistency) $\widehat{\psi}_{DR}(a) \stackrel{P}{\longrightarrow} \psi(a)$.
    \item (Asymptotic distribution) $\widehat{\psi}_{DR}(a)$ has the asymptotic representation
    $$
    \sqrt{n}\left(\widehat{\psi}_{DR}(a)-\psi(a)\right)=\sqrt{n}\left(\frac{1}{n} \sum_{i=1}^n H\left(e^*_a(X_i), h^*_a(X_i)\right)-\mathrm{E}\left[H\left(e^*_a(X), h^*_a(X)\right)\right]\right)+R e+o_P(1),
    $$
    where
    $$
    R e \leq \sqrt{n} O_P\left(\left\|\widehat{h}_a(X)-\mathrm{E}\left[L\left(Y, \mu(X^*)\right) \mid X, A=a\right]\right\|_2^2 \times\Big\|\widehat{e}_a(X)-\operatorname{Pr}[A=a \mid X]\Big\|_2^2\right) 
    $$
    and thus if $\widehat{h}_a(X)$ and $\widehat{e}_a(X)$ converge at combined rate of at least $\sqrt{n}$ then
    $$
    \sqrt{n}\left(\widehat{\psi}_{DR}(a)-\psi(a)\right) \stackrel{d}{\longrightarrow} N\left(0, \operatorname{Var}\left[H(e^*_a(X), h^*_a(X))\right]\right)
    $$
\end{enumerate}

\subsubsection{Consistency}
Using the probability limits $e_a^*(X)$ and $h_a^*(X)$ defined previously, the double robust estimator $\widehat{\psi}_{DR}(a)$ converges in probability to 
$$\widehat{\psi}_{DR}(a) \stackrel{P}{\longrightarrow} \E\left[h^*_a(X)+\frac{I(A = a)}{e^*_a(X)}\left(L\left(Y, \mu\left(X^*\right)\right)-h^*_a(X)\right)\right]$$
Here we show that the right-hand side is equal to $\psi$ under assumptions D1-
D4 when either:
\begin{enumerate}
    \item $\widehat{e}_a(X)$ is correctly specified
    \item $\widehat{h}_a(X)$ is correctly specified
\end{enumerate}
First consider the case where $\widehat{e}_a(X)$ is correctly specified, that is $e^*_a(X)=\operatorname{Pr}[A=a \mid X]$, but we do not assume that the limit $h^*_a(X)$ is equal to $\left.\mathrm{E}\left[L\left(Y, g\left(X^*\right)\right) \mid X, A=a\right]\right)$. Recall, as shown previously $\psi = \E\left[\frac{I(A = a)}{\Pr[A = a \mid X]}L(Y, \widehat{\mu}(X^*))\right] $
$$
\begin{aligned}
\widehat{\psi}_{DR}(a) & \stackrel{P}{\rightarrow}  \E\left[h^*_a(X)+\frac{I(A = a)}{e^*_a(X)}\left(L\left(Y, \mu\left(X^*\right)\right)-h^*_a(X)\right)\right] \\
& =\E\left[h^*_a(X)-\frac{I(A = a)}{e^*_a(X)}h^*_a(X)\right]+\psi(a) \\
& =\E\left[\E\left[h^*_a(X)-\frac{I(A = a)}{e^*_a(X)}h^*_a(X) \mid X \right]\right]+\psi(a) \\
& =\E\left[h^*_a(X)-\frac{1}{e^*_a(X)}h^*_a(X) \E\left[I(A = a) \mid X \right]\right]+\psi(a) \\
& =\E\left[h^*_a(X)-\frac{1}{e^*_a(X)}h^*_a(X) \Pr\left[A = a \mid X \right]\right]+\psi(a) \\
& =\E\left[h^*_a(X)-h^*_a(X)\right]+\psi(a) \\
& =\psi(a) .
\end{aligned}
$$
Next consider the case when $\widehat{h}_a(X)$ is correctly specified, that is
$$
h^*_a(X)=\mathrm{E}\left[L\left(Y, g\left(X^*\right)\right) \mid X, A=a\right]
$$
and this time we do not make the assumptions that the limit $e^*_a(X)$ is equal to $\operatorname{Pr}[A=a \mid X]$. Recall, as shown previously $\psi(a) = \E\left[\E\left[L(Y, \widehat{\mu}(X^*))\mid X, A = a\right]\right] $. 

$$
\begin{aligned}
\widehat{\psi}_{DR}(a) & \stackrel{P}{\rightarrow}  \E\left[h^*_a(X)+\frac{I(A = a)}{e^*_a(X)}\left(L\left(Y, \mu\left(X^*\right)\right)-h^*_a(X)\right)\right] \\
&= \E\left[h^*_a(X)\right]+\E\left[\frac{I(A = a)}{e^*_a(X)}\left(L\left(Y, \mu\left(X^*\right)\right)-h^*_a(X)\right)\right] \\
& =\psi(a)+\E\left[\frac{I(A = a)}{e^*_a(X)}\left(L\left(Y, \mu\left(X^*\right)\right)-h^*_a(X)\right)\right] \\
& =\psi(a)+\E\left[\E\left[\frac{I(A = a)}{e^*_a(X)}\left(L\left(Y, \mu\left(X^*\right)\right)-h^*_a(X)\right) \mid X \right]\right] \\
& =\psi(a)+\E\left[\frac{I(A = a)}{e^*_a(X)} \E\left[\left(L\left(Y, \mu\left(X^*\right)\right)-h^*_a(X)\right) \mid X \right]\right] \\
& =\psi(a)+\E\left[\E\left[\left(L\left(Y, \mu\left(X^*\right)\right)-h^*_a(X)\right) \mid X, A=a \right]\right] \\
& =\psi(a)+\E\left[\E\left[L\left(Y, \mu\left(X^*\right)\right) \mid X, A=a  \right] -h^*_a(X) \right] \\
& =\psi(a)+\E\left[h^*_a(X) -h^*_a(X) \right] \\
& =\psi(a) .
\end{aligned}
$$

\subsubsection{Asymptotic distribution}

For a random variable $W$ we define notation
$$
\mathbb{G}_n(W)=\sqrt{n}\left(\frac{1}{n} \sum_{i=1}^n W_i-\mathrm{E}[W]\right) .
$$
and thus the asymptotic representation of $\widehat{\psi}_{DR}(a)$ can be written
$$
\begin{aligned}
\sqrt{n}\left(\widehat{\psi}_{DR}(a)-\psi(a)\right)=\mathbb{G}_n &(H(\widehat{e}_a(X), \widehat{h}_a(X)))-\mathbb{G}_n\left(H\left(e^*_a(X), h^*_a(X)\right)\right) \\
& +\mathbb{G}_n\left(H\left(e^*_a(X), h^*_a(X)\right)\right) \\
& +\sqrt{n}(\mathrm{E}[H(\widehat{e}_a(X), \widehat{h}_a(X))]-\psi(a))
\end{aligned}
$$
where we add and subtract the term $\mathbb{G}_n\left(H\left(e^*_a(X), h^*_a(X)\right)\right)$ and add another zero term in \\$\sqrt{n}(\mathrm{E}[H(\widehat{e}_a(X), \widehat{h}_a(X))]-\psi(a))$. For the first term, Assumption D1 implies
$$
\mathbb{G}_n(H(\widehat{e}_a(X), \widehat{h}_a(X)))-\mathbb{G}_n\left(H\left(e^*_a(X), h^*_a(X)\right)\right)=o_P(1)
$$
Let 
$$
Re=\sqrt{n}(\mathrm{E}[H(\widehat{e}_a(X), \widehat{h}_a(X))]-\psi(a))
$$
now we have
$$
\sqrt{n}\left(\widehat{\psi}_{DR}(a)-\psi(a)\right)=\sqrt{n}\left(\frac{1}{n} \sum_{i=1}^n\left(H\left(e^*_a(X_i), h^*_a(X_i)\right)-\mathrm{E}\left[H\left(e^*_a(X), h^*_a(X)\right)\right]\right)\right)+R e+o_P(1)
$$
Let's try to calculate the upper bound of $Re$. First, note

$$
n^{-1 / 2} Re=
\underbrace{\mathrm{E}\left[\widehat{h}_a(X)\right]}_{R_1}+\underbrace{\mathrm{E}\left[\frac{I(A = a)}{\widehat{e}_a(X)}\left[L\left(Y, \mu\left(X^*\right)\right)-\widehat{h}_a(X)\right]\right]}_{R_2}-\psi(a).
$$
We rewrite term $R_2$ as:
$$
\begin{aligned}
R_2 & =\mathrm{E}\left[\frac{ I(A=a)}{\widehat{e}_a(X)}\left\{L\left(Y, \mu\left(X^*\right)\right)-\widehat{h}_a(X)\right\}\right] \\
& =\mathrm{E}\left[\mathrm{E}\left[\frac{ I(A=a)}{\widehat{e}_a(X)}\left\{L\left(Y, \mu\left(X^*\right)\right)-\widehat{h}_a(X)\right\} \mid X\right]\right] \\
& =\mathrm{E}\left[\frac{1}{\widehat{e}_a(X)} \mathrm{E}\left[\frac{I(A=a)}{\operatorname{Pr}[A=a \mid X]} \operatorname{Pr}[A=a \mid X]\left\{L\left(Y, \mu\left(X^*\right)\right)-\widehat{h}_a(X)\right\} \mid X\right]\right] \\
& =\mathrm{E}\left[\frac{1}{\widehat{e}_a(X)} \mathrm{E}\left[\operatorname{Pr}[A=a \mid X]\left\{L\left(Y, \mu\left(X^*\right)\right)-\widehat{h}_a(X)\right\} \mid X, A=a\right]\right] \\
& =\mathrm{E}\left[\frac{1}{\widehat{e}_a(X)} \operatorname{Pr}[A=a \mid X]\left\{\mathrm{E}\left[L\left(Y, \mu\left(X^*\right)\right) \mid X, A=a\right]-\widehat{h}_a(X)\right\}\right]
\end{aligned}
$$
Combining the above gives
$$
\begin{aligned}
n^{-1 / 2} R e & =\mathrm{E}\left[\widehat{h}_a(X)\right]+\mathrm{E}\left[\frac{I(A = a)}{e_a^{\prime}(X)}\left[L\left(Y, \mu\left(X^*\right)\right)-h_a^{\prime}(X)\right]\right]-\psi(a) \\
& =\mathrm{E}\left[\widehat{h}_a(X)\right]+\mathrm{E}\left[\frac{1}{\widehat{e}_a(X)} \operatorname{Pr}[A=a \mid X]\left\{\mathrm{E}\left[L\left(Y, \mu\left(X^*\right)\right) \mid X, A=a\right]-\widehat{h}_a(X)\right\}\right] \\ 
&- \E\left[\E\left[L(Y, \widehat{\mu}(X^*))\mid X, A = a\right]\right]\\
& =\E\left[\left\{\mathrm{E}\left[L\left(Y, \mu\left(X^*\right)\right) \mid X, A=a\right]-\widehat{h}_a(X)\right\}\times\left\{\frac{1}{\widehat{e}_a(X)} \operatorname{Pr}[A=a \mid X]-1\right\}\right]
\end{aligned}
$$
Using the Cauchy-Schwartz inequality we get.
$$
\begin{aligned}
Re & \leq \sqrt{n}\left(\mathrm{E}\left[\left\{\mathrm{E}\left[L\left(Y, \mu\left(X^*\right)\right) \mid X, A=a\right]-\widehat{h}_a(X)\right\}^2\right]\right)^{1 / 2} \\
& \times\left(\mathrm{E}\left[\left\{\frac{1}{\widehat{e}_a(X)} \operatorname{Pr}[A=a \mid X]-1\right\}^2\right]\right)^{1 / 2} \\
\leq & \sqrt{n} O_P\left(\left\|\mathrm{E}\left[L\left(Y, \mu\left(X^*\right)\right) \mid X, A=a\right]-\widehat{h}_a(X)\right\|_2^2 \times\Big\|\widehat{e}_a(X)-\operatorname{Pr}[A=a \mid X]\Big\|_2^2\right)
\end{aligned}
$$

If both models $\widehat{e}_a(X)$ and $\widehat{h}_a(X)$ are correctly specified and converge at a combined rate faster than $\sqrt{n}$, then $R e=o_P(1)$ and
$$
\begin{aligned}
\sqrt{n}\left(\widehat{\psi}_{DR}(a)-\psi(a)\right) & =\sqrt{n}\left(\frac{1}{n} \sum_{i=1}^n H\left(\operatorname{Pr}\left[A=a \mid X_i\right], \mathrm{E}\left[L\left(Y, g\left(X^*\right)\right) \mid A=a, X_i\right]\right)\right. \\
& \left.-\mathrm{E}\left[H\left(\operatorname{Pr}[A=a \mid X], \mathrm{E}\left[L\left(Y, g\left(X^*\right)\right) \mid A=a, X\right]\right)\right]\right)+o_P(1)
\end{aligned}
$$
By the central limit theorem,
$$
\sqrt{n}\left(\frac{1}{n} \sum_{i=1}^n H\left(e^*_a(X_i), h^*_a(X_i)\right)-\mathrm{E}\left[H\left(e^*_a(X_i), h^*_a(X_i)\right)\right]\right) \stackrel{d}{\longrightarrow} N\left(0, \operatorname{Var}\left[H\left(e^*_a(X), h^*_a(X)\right)\right]\right)
$$
completing the proof.
\newpage 
\section{Risk calibration curve}\label{sec:calib}
When the outcome is binary, another common metric of model performance is risk calibration. Calibration is a measure of the relibability of the risk estimates produced by the fitted model $\widehat{\mu}(X^*)$. For instance, among a sample of patients who receive a risk prediction of 17\% does the outcome really occur for roughly 17\% of them over the follow up period? This can be nonparametrically evalutated across a range of risks by estimating the so-called ``calibration'' curve, i.e. the observed risk as a function of the predicted risk. For counterfactual predictions the relevant calibration curve is the counterfactual risk that would be observed under intervetion $A=a$ as a function of the predicted risk, or
\begin{equation}\label{eqn:calib_estimand}
    \psi(a) = \E[I(Y^a = 1) \mid \widehat{\mu}(x^*)].
\end{equation}

\subsection{Identification}
Here we show that the counterfactual calibration curve $\psi(a) = \E[I(Y^a = 1) \mid \widehat{\mu}(x^*)]$ is identified using the observed data under the assumptions of section \ref{sec:identifiability}.
\begin{theorem}
     Under conditions A1-A3, the risk calibration curve is identified by the observed data functionals
\begin{equation}\label{eqn:cl_calib_estimand}
    \psi(a) = \E[\E\{I(Y = 1) \mid X, A = a, \widehat{\mu}(x^*), D_{test} = 1\}\mid \widehat{\mu}(x^*), D_{test} = 1]
\end{equation}
and 
\begin{equation}\label{eqn:ipw_calib_estimand}
    \psi(a) = \E\left[\frac{I(A = a)}{\Pr[A = a \mid X, \widehat{\mu}(x^*), D_{test} = 1]} I(Y=1) \mid \widehat{\mu}(x^*), D_{test} = 1\right]
\end{equation}
in the test set. 

\end{theorem}

\begin{proof}
    For the first representation we have 
    \begin{align*}
        \psi(a) &= \E[I(Y^a = 1) \mid \widehat{\mu}(x^*)] \\
        & = \E[I(Y^a = 1) \mid \widehat{\mu}(x^*), D_{test} = 1] \\
        & = \E[\E\{I(Y^a = 1) \mid X, \widehat{\mu}(x^*), D_{test} = 1\}\mid \widehat{\mu}(x^*), D_{test} = 1] \\
        & = \E[\E\{I(Y^a = 1) \mid X, A = a, \widehat{\mu}(x^*), D_{test} = 1\}\mid \widehat{\mu}(x^*), D_{test} = 1] \\
        & = \E[\E\{I(Y = 1) \mid X, A = a, \widehat{\mu}(x^*), D_{test} = 1\}\mid \widehat{\mu}(x^*), D_{test} = 1]
    \end{align*}
    where the first line follows from the definition of $\psi(a)$, the second from random sampling of the test set, the third from the law of iterated expectations, the fourth from the exchangeability condition A1, and the fifth from the consistency condition A2. Recall that $X^*$ is a subset of $X$. For the second representation, we show that it is equivalent to the first 
    \begin{align*}
        \psi(a) &= \E[\E\{I(Y = 1) \mid X, A = a, \widehat{\mu}(x^*), D_{test} = 1\}\mid \widehat{\mu}(x^*), D_{test} = 1] \\
        &= \E\left[\E\left\{\frac{I(A = a)}{\Pr[A = a \mid X, \widehat{\mu}(x^*), D_{test} = 1]} I(Y=1) \mid X, \widehat{\mu}(x^*), D_{test} = 1\right\}\mid \widehat{\mu}(x^*), D_{test} = 1\right] \\
        &= \E\left[\frac{I(A = a)}{\Pr[A = a \mid X, \widehat{\mu}(x^*), D_{test} = 1]} I(Y=1) \mid \widehat{\mu}(x^*), D_{test} = 1\right]
    \end{align*}
    where the second line follows from the definition of conditional expectation, and the last reverses the law of iterated expectations.
\end{proof}

\subsection{Estimation}
Unlike previous sections, estimation of the full risk calibration curve using sample analogs of the identified expressions \ref{eqn:cl_calib_estimand} and \ref{eqn:ipw_calib_estimand} is generally infeasible because they are conditional on a continuous risk score. Instead analysts typically perform either kernel or binned estimation of the calibration curve functional. In the case of the counterfactual risk calibration curve under a hypothetical intervention, the expression above suggests modifying these approaches either through the use of inverse probability weights or an outcome model.

\section{Area under ROC curve}\label{sec:auc}
Another common metric for the performance of a risk prediction model $\mu(X^*)$ is the area under the receiver operating characteristic (ROC) curve, often referred to as simply the area under the curve (AUC). The AUC can be interpreted as the probability that a randomly sampled observation with the outcome has a higher predicted value than a randomly sampled observation without the outcome. In that sense, it is a measure of the discriminative ability of the model, i.e. the ability to distinguish between cases and noncases. For counterfactual predictions the relevant AUC is the counterfactual AUC that would be observed under intervetion $A=a$, or
\begin{equation}\label{eqn:auc_estimand}
    \psi(a) = \E[I\left(\widehat{\mu}(X^*_i) > \widehat{\mu}(X^*_j)\right) \mid Y_i^a = 1, Y_j^a = 0].
\end{equation}

\subsection{Identification}
Here, we show that the counterfactual AUC $\psi(a)$ is identified by the observed data under a modified set of identification conditions, namely:
\begin{enumerate}
    \item[E1.] \textit{Exchangeability.} $Y^a \perp\!\!\!\perp A \mid X$ 
    \item[E2.] \textit{Consistency.} $Y^a = Y$ if $A = a$
    \item[E3.] \textit{Positivity.} (i) $\Pr(A = a | X = x) > 0$ for all $x$ that have positive density in $f(X, A = a)$, (ii) $\mathrm{E}\left[\Pr[Y = 1 | X_i, A = a]\Pr[Y = 0 | X_j, A = a]\right] > 0 $, where $i$ is a random observation that has the outcome and $j$ is random observation without the outcome.
\end{enumerate}

\begin{theorem}
    Under conditions E1-E3, the counterfactual AUC is identified by the observed data functionals in the test set
\begin{equation}\label{eqn:cl_auc_estimand}
    \psi(a) = \frac{\mathrm{E}\left[I\left(\widehat{\mu}(X_i^*)>\widehat{\mu}(X_j^*)\right)m_a(X_i, X_j) \right]}{\mathrm{E}\left[m_a(X_i, X_j) \right]} 
\end{equation}
and 
\begin{equation}\label{eqn:ipw_auc_estimand}
    \psi(a) = \frac{\mathrm{E}\left[\frac{I\left(\widehat{\mu}(X_i^*)>\widehat{\mu}(X_j^*), Y_i=1, Y_j=0, A_i = a, A_j = a\right)}{\pi_a(X_i, X_j)} \right]}{\mathrm{E}\left[\frac{I\left(Y_i=1, Y_j=0, A_i = a, A_j = a\right)}{\pi_a(X_i, X_j)}\right]} 
\end{equation}
where the subscripts $i$ and $j$ denote a random pair of observations from the test set where
\begin{equation*}
    m_a(X_i, X_j) = \operatorname{Pr}\left[Y_i=1 \mid X_i,A_i = a, D_{test,i} = 1\right] \Pr\left[ Y_j=0 \mid X_j, A_j = a, D_{test,j} = 1\right]
\end{equation*}
and 
\begin{equation*}
    \pi_a(X_i, X_j) = \Pr\left[A_i = a \mid X_i, D_{test,i} = 1\right] \Pr\left[A_j = a  \mid X_j, D_{test,j} = 1\right]
\end{equation*}
for a pair of covariate vectors $X_i$ and $X_j$. 
\end{theorem}

\begin{proof}
    For the first representation we have 
$$
\begin{aligned}
 \psi(a)&= \mathrm{E}\left[I\left(\widehat{\mu}(X_i^*)>\widehat{\mu}(X_j^*)\right) \mid Y^a_i=1, Y^a_j=0\right] \\
&= \frac{\mathrm{E}\left[I\left(\widehat{\mu}(X_i^*)>\widehat{\mu}(X_j^*), Y^a_i=1, Y^a_j=0\right)\right]}{\operatorname{Pr}\left[Y^a_i=1, Y^a_j=0\right]} \\
&= \frac{\mathrm{E}\left[\mathrm{E}\left[I\left(\widehat{\mu}(X_i^*)>\widehat{\mu}(X_j^*), Y^a_i=1, Y^a_j=0\right) \mid X_i, X_j\right]\right]}{\mathrm{E}\left[\operatorname{Pr}\left[Y^a_i=1, Y^a_j=0 \mid X_i, X_j\right]\right]} \\
&= \frac{\mathrm{E}\left[I\left(\widehat{\mu}(X_i^*)>\widehat{\mu}(X_j^*)\right) \operatorname{Pr}\left[Y^a_i=1, Y^a_j=0  \mid X_i, X_j\right]\right]}{\mathrm{E}\left[\operatorname{Pr}\left[Y^a_i=1, Y^a_j=0  \mid X_i, X_j\right]\right]} \\
&= \frac{\mathrm{E}\left[I\left(\widehat{\mu}(X_i^*)>\widehat{\mu}(X_j^*)\right) \operatorname{Pr}\left[Y^a_i=1, Y^a_j=0  \mid X_i, X_j, A_i = a, A_j = a\right]\right]}{\mathrm{E}\left[\operatorname{Pr}\left[Y^a_i=1, Y^a_j=0  \mid A_i = a, A_j = a, X_i, X_j\right]\right]} \\
&= \frac{\mathrm{E}\left[I\left(\widehat{\mu}(X_i^*)>\widehat{\mu}(X_j^*)\right) \operatorname{Pr}\left[Y^a_i=1 \mid X_i,A_i = a\right] \Pr\left[ Y^a_j=0 \mid X_j, A_j = a\right]\right]}{\mathrm{E}\left[\operatorname{Pr}\left[Y^a_i=1 \mid X_i,A_i = a\right] \Pr\left[ Y^a_j=0 \mid X_j, A_j = a\right]\right]} \\
&= \frac{\mathrm{E}\left[I\left(\widehat{\mu}(X_i^*)>\widehat{\mu}(X_j^*)\right) \operatorname{Pr}\left[Y_i=1 \mid X_i,A_i = a\right] \Pr\left[ Y_j=0 \mid X_j, A_j = a\right]\right]}{\mathrm{E}\left[\operatorname{Pr}\left[Y_i=1 \mid X_i,A_i = a\right] \Pr\left[ Y_j=0 \mid X_j, A_j = a\right]\right]} \\
&= \frac{\mathrm{E}\left[I\left(\widehat{\mu}(X_i^*)>\widehat{\mu}(X_j^*)\right) \operatorname{Pr}\left[Y_i=1 \mid X_i,A_i = a\right] \Pr\left[ Y_j=0 \mid X_j, A_j = a\right]\right]}{\mathrm{E}\left[\operatorname{Pr}\left[Y_i=1 \mid X_i,A_i = a\right] \Pr\left[ Y_j=0 \mid X_j, A_j = a\right]\right]} \\
&= \frac{\mathrm{E}\left[I\left(\widehat{\mu}(X_i^*)>\widehat{\mu}(X_j^*)\right)m_a(X_i, X_j) \right]}{\mathrm{E}\left[m_a(X_i, X_j) \right]}
\end{aligned}
$$

where the first line follows from the definition of $\psi(a)$, the second from the definition of conditional probability, the third from the law of iterated expectations, the fourth from the definition of conditional expectation, the fifth from the exchangeability condition E1, the sixth from independence of potential outcomes, the seventh from the consistency condition E2, the eighth from random sampling of the test set, and the ninth applies the definition of $m_a(X_i, X_j)$. Recall that $X^*$ is a subset of $X$. For the second representation, we will show that it is equivalent to the first. Starting from line five above

$$
\begin{aligned}
\psi(a)&= \frac{\mathrm{E}\left[I\left(\widehat{\mu}(X_i^*)>\widehat{\mu}(X_j^*)\right) \operatorname{Pr}\left[Y^a_i=1, Y^a_j=0  \mid X_i, X_j, A_i = a, A_j = a\right]\right]}{\mathrm{E}\left[\operatorname{Pr}\left[Y^a_i=1, Y^a_j=0  \mid A_i = a, A_j = a, X_i, X_j\right]\right]} \\
&= \frac{\mathrm{E}\left[I\left(\widehat{\mu}(X_i^*)>\widehat{\mu}(X_j^*)\right) \operatorname{Pr}\left[Y_i=1, Y_j=0  \mid X_i, X_j, A_i = a, A_j = a\right]\right]}{\mathrm{E}\left[\operatorname{Pr}\left[Y_i=1, Y_j=0  \mid A_i = a, A_j = a, X_i, X_j\right]\right]} \\
&= \frac{\mathrm{E}\left[I\left(\widehat{\mu}(X_i^*)>\widehat{\mu}(X_j^*)\right) \frac{\operatorname{Pr}\left[Y_i=1, Y_j=0, A_i = a, A_j = a \mid X_i, X_j\right]}{\operatorname{Pr}\left[A_i = a, A_j = a \mid X_i, X_j\right]}\right]}{\mathrm{E}\left[\frac{\operatorname{Pr}\left[Y_i=1, Y_j=0, A_i = a, A_j = a \mid X_i, X_j\right]}{\operatorname{Pr}\left[A_i = a, A_j = a \mid X_i, X_j\right]}\right]} \\
&= \frac{\mathrm{E}\left[\mathrm{E}\left[I\left(\widehat{\mu}(X_i^*)>\widehat{\mu}(X_j^*)\right) \frac{\operatorname{Pr}\left[Y_i=1, Y_j=0, A_i = a, A_j = a \mid X_i, X_j\right]}{\operatorname{Pr}\left[A_i = a, A_j = a \mid X_i, X_j\right]} \mid X_i, X_j\right]\right]}{\mathrm{E}\left[\mathrm{E}\left[\frac{\operatorname{Pr}\left[Y_i=1, Y_j=0, A_i = a, A_j = a \mid X_i, X_j\right]}{\operatorname{Pr}\left[A_i=a, A_j = a \mid X_i, X_j\right]} \mid X_i, X_j\right]\right]} \\
&= \frac{\mathrm{E}\left[\frac{I\left(\widehat{\mu}(X_i^*)>\widehat{\mu}(X_j^*)\right)}{\operatorname{Pr}\left[A_i = a \mid X_i\right] \operatorname{Pr}\left[A_j = a \mid X_j\right]} \operatorname{Pr}\left[Y_i=1, Y_j=0, A_i = a, A_j = a \mid X_i, X_j\right]\right]}{\mathrm{E}\left[\frac{\operatorname{Pr}\left[Y_i=1, Y_j=0, A_i = a, A_j = a \mid X_i, X_j\right]}{\operatorname{Pr}[A_i = a \mid X_i] \Pr[A_j = a \mid X_j]}\right]} \\
& = \frac{\mathrm{E}\left[\frac{I\left(\widehat{\mu}(X_i^*)>\widehat{\mu}(X_j^*), Y_i=1, Y_j=0, A_i = a, A_j = a\right)}{\operatorname{Pr}\left[A_i = a \mid X_i\right] \operatorname{Pr}\left[A_j = a \mid X_j\right]} \right]}{\mathrm{E}\left[\frac{I\left(Y_i=1, Y_j=0, A_i = a, A_j = a\right)}{\operatorname{Pr}[A_i = a \mid X_i] \Pr[A_j = a \mid X_j]}\right]} \\
& = \frac{\mathrm{E}\left[\frac{I\left(\widehat{\mu}(X_i^*)>\widehat{\mu}(X_j^*), Y_i=1, Y_j=0, A_i = a, A_j = a\right)}{\pi_a(X_i, X_j)} \right]}{\mathrm{E}\left[\frac{I\left(Y_i=1, Y_j=0, A_i = a, A_j = a\right)}{\pi_a(X_i, X_j)}\right]} 
\end{aligned}
$$

where the second line follows from consistency $E2$, the third from the definition of conditional probability, the fourth from iterated expectations, the fifth removes the constant fraction outside expectation, the sixth reverses the law of iterated expectations and the last applies random sampling of the test set and the definition of $\pi_a(X_i, X_j)$.
\end{proof}

\subsection{Plug-in estimation}
Using sample analogs for the identified expressions \ref{eqn:cl_auc_estimand} and \ref{eqn:ipw_auc_estimand}, we obtain two plug-in estimators for the counterfactual AUC
    \begin{equation*}
        \widehat{\psi}_{OM}(a) = \frac{\sum_{i \neq j}^n\widehat{h}_a(X_i) (1 - \widehat{h}_a(X_j)) I(\widehat{\mu}(X_i^*)>\widehat{\mu}(X_j^*), D_{test, i} = 1,  D_{test, j} = 1) }{\sum_{i \neq j}^n\widehat{h}_a(X_i) (1 - \widehat{h}_a(X_j)) I(D_{test, i} = 1,  D_{test, j} = 1)}
    \end{equation*}
    and 
    \begin{equation*}
        \widehat{\psi}_{IPW}(a) = \frac{\mathlarger{\mathlarger{\sum}}_{i \neq j}^n \dfrac{I\left(\widehat{\mu}(X_i^*)>\widehat{\mu}(X_j^*), Y_i = 1, Y_j = 0, A_i = a, A_j = a, D_{test, i} = 1,  D_{test, j} = 1\right)}{\widehat{e}_a(X_i) \widehat{e}_a(X_j)}}{\mathlarger{\mathlarger{\sum}}_{i \neq j}^n\dfrac{I\left(Y_i = 1, Y_j = 0, A_i = a, A_j = a, D_{test, i} = 1,  D_{test, j} = 1\right)}{\widehat{e}_a(X_i) \widehat{e}_a(X_j)}}
    \end{equation*}
    where $\widehat{m}_a(X_i, X_j) = \widehat{h}_a(X_i) (1 - \widehat{h}_a(X_j))$ and $\widehat{h}_a$ is an estimator for $\operatorname{Pr}[Y_i=1 | X_i,A_i = a, D_{test,i} = 1]$ and where $\widehat{\pi}_a(X_i, X_j) = \widehat{e}_a(X_i) \widehat{e}_a(X_j)$ and $\widehat{e}_a$ is an estimator for $\Pr[A_i = a | X_i, D_{test,i} = 1]$. Here, we call the first plug-in estimator the outcome model estimator $ \widehat{\psi}_{OM}$ and the second the inverse probability weighted estimator $\widehat{\psi}_{IPW}$. 

\newpage
\section{Additional application details}
The Multi-Ethnic Study on Atherosclerosis (MESA) study is a population-based sample of 6,814 men and women aged 45 to 84 drawn from six communities (Baltimore; Chicago; Forsyth County, North Carolina; Los Angeles; New York; and St. Paul, Minnesota) in the United States between 2000 and 2002. The sampling procedure, design, and methods of the study have been described previously \cite{bild_multi-ethnic_2002}. Study teams conducted five examination visits between 2000 and 2011 in 18 to 24 month intervals focused on the prevalence, correlates, and progression of subclinical cardiovascular disease. These examinations included assessments of lipid-lowering (primarily statins) and other medication use as well as cardiovascular risk factors such as systolic blood pressure, serum cholesterol, cigarette smoking, height, weight, and diabetes. 

Our goal was to emulate a single-arm trial corresponding to the AHA guidelines on initiation of statin therapy for primary prevention of cardiovascular disease in the MESA cohort and use the emulated trial to develop a prediction model for the treatment-naive risk. The AHA guidelines stipulate that patients aged 40 to 75 with serum LDL cholesterol levels between 70 mg/dL and 190 mg/dL and no history of cardiovascular disease should initiate statins if their risk exceeds 7.5\%. Therefore, we considered MESA participants who completed the baseline examination, had no recent history of statin use, no history of cardiovascular disease, and who met the criteria described in the guidelines (excluding the risk threshold) as eligible to participate in the trial. The primary endpoint was time to atherosclerotic cardiovascular disease (ASCVD), defined as nonfatal myocardial infarction, coronary heart disease death, or ischemic stroke. 

Follow up began at the second examination cycle to enable a ``wash out'' period for statin use and to ensure adequate pre-treatment covariates to control confouding. We constructed a sequence of nested trials starting at each examination cycle from exam 2 through exam 5 and pooled the results from all 4 trials into a single analysis and used a robust variance estimator to account for correlation among duplicated participants. In each nested trial, we used the corresponding questionnaire to determine eligibility as well as statin initiators versus non-initiators. Because the exact timing of statin initiation was not known with precision, in each trial, we estimated the start of follow up for initiators and non-initators by drawing a random month between their current and previous examinations. We explored alternative definitions of the start of follow up in sensitivity analyses in the appendix. To mimic the targeted single-arm trial we limited to non-initiators for development of the prediction models.

\subsection{Propensity score models}\label{sec:covs}
In the emulated single arm trial, statin initiation can be viewed as ``non-adherence'' which can be adjusted for by inverse probability weighting, therefore we censored participants when they initiated statins. To calculate the weights, we estimated two logistic regression models: one for the probability of remaining untreated given past covariate history (denominator model) and one for probability of remaining untreated given the selected baseline predictors (numerator model). In the denominator model we included the following covariates:
\begin{itemize}
    \item \textit{Demographic factors} - Age, gender, marital status, education, race/ethnicity, employment, health insurance status, depression, perceived discrimination, emotional support, anger and anxiety scales, and neighborhood score.
    \item \textit{Risk factors} - Systolic and diastolic blood pressure, serum cholesterol levels (LDL, HDL, Triglycerides), hypertension, diabetes, waist circumference, smoking, alcohol consumption, exercise, family history of CVD, calcium score, hypertrophy on ECG, CRP, IL-6, number of pregnancies, oral contraceptive use, age of menopause.
    \item \textit{Medication use} - Anti-hypertensive use, insulin use, daily aspirin use, anti-depressant use, vasodilator use, anti-arryhtmic use. 
\end{itemize}
Time-varying demographic factors and risk factors were lagged such that values from the previous examination cycle were used. 


\end{appendix}

\onehalfspacing

\end{document}